%% file: UniformHashing.tex
\def\draft{0} 
\def\CR{0} 

\documentclass[11pt]{article}
\usepackage{amsfonts,amsmath,amssymb,amsthm,boxedminipage,color,url,fullpage}
\usepackage{enumerate}
\usepackage{enumitem}
\usepackage[numbers]{natbib} 
\usepackage[pdfstartview=FitH,colorlinks,linkcolor=blue,filecolor=blue,citecolor=blue,urlcolor=blue]{hyperref} 
\usepackage[labelfont=bf]{caption}
\usepackage{subcaption}
\usepackage{aliascnt}
\usepackage{cleveref}
\usepackage{xspace}
\usepackage{graphicx}

\input{macros}

\newcommand{\PP}{\mathcal{PP}}
\newcommand{\ADW}{\mathcal{ADW}}
\newcommand{\adw}{\MathAlg{adw}}

\newcommand{\GGM}{\ensuremath{\mathcal{GGM}}\xspace}
\newcommand{\JPTCon}{\ensuremath{\mathcal{JPT}}\xspace}

\newcommand{\ab}{{\overline{a}}}
\newcommand{\qb}{{\overline{q}}}
\newcommand{\cd}{{\cal{D}}}

\newcommand{\Ab}{{\overline{A}}}
\newcommand{\Qb}{{\overline{Q}}}
\newcommand{\Atb}{{\overline{A'}}}
\newcommand{\Qtb}{{\overline{Q'}}}
\newcommand{\mD}{\MathAlg{D}}

\newcommand{\hcD}{{\widehat{\Dc}}}
\newcommand{\hcT}{{\widetilde{\Dc}}}
\newcommand{\vb}{{\overline{v}}}
\newcommand{\wb}{{\overline{w}}}
\newcommand{\RandFunc}{{\F(\cU,\cV)}}
\newcommand{\unbounded}{\infty}

\newcommand{\cc}{z}
\newcommand{\adwparam}{z}
\newcommand{\m}{m}
\newcommand{\mv}{\overline{\m}}
\newcommand{\yv}{\overline{y}}
\newcommand{\fv}{\overline{f}}
\newcommand{\piv}{\overline{\pi}}
\newcommand{\adww}[1]{a_{#1}}

\newcommand{\g}{g}
\newcommand{\gv}{\overline{\g}}

\newcommand{\rand}{\mathsf{r}}
\newcommand{\eval}{\mathsf{e}}

\newcommand{\evalManyWise}[3]{\eval^{(#1)}}

\title{Hardness-Preserving Reductions via Cuckoo Hashing\footnote{This is the final draft of this paper. The full
		version was published in the Journal of Cryptology 
		\cite{BermanHKN19}. An extended abstract of this work appeared in the Theory of Cryptography Conference (TCC) 2013 \cite{BermanHKN13}.}\Draft{\\{\small \sc Working Draft: Please Do Not Distribute}}}

\author{ Itay Berman\thanks{MIT Computer Science and Artificial Intelligence
    Laboratory (CSAIL). Email: \texttt{itayberm@mit.edu}. Most of this work was
    done while the author was in the School of Computer Science, Tel Aviv
    University. Supported in part by NSF Grants CNS-1413920 and CNS-1350619, and
    by the Defense Advanced Research Projects Agency (DARPA) and the U.S. Army
    Research Office under contracts W911NF-15-C-0226 and
    W911NF-15-C-0236. }~\footnotemark[4] \and Iftach Haitner\thanks{School of
    Computer Science, Tel Aviv University. Email:
    \texttt{iftachh@tauex.tau.ac.il}.}~\thanks{Research supported in part by Check
    Point Institute for Information Security and the (I-CORE) program (Center
    No.\ 4/11) of the Planning and Budgeting Committee. Preparation of the journal version was supported by  by ERC starting grant 638121.}  \and Ilan
  Komargodski\thanks{Cornell Tech, New York, USA. Email:
    \texttt{komargodski@cornell.edu}. Supported in part by a Packard Foundation
    Fellowship. Most of this work done while being a Ph.D student at the
    Weizmann Institute of Science, supported in part by a grant from the I-CORE
    Program (Center No.\ 4/11) of the Planning and Budgeting Committee and the
    Israel Science Foundation.} \and Moni Naor\thanks{Department of Computer
    Science and Applied Mathematics, Weizmann Institute of Science, Rehovot
    76100, Israel. Email: \texttt{moni.naor@weizmann.ac.il}. Research supported
    in part by a grant from the I-CORE Program (Center No.\ 4/11) of the
    Planning and Budgeting Committee and the Israel Science
    Foundation. Incumbent of the Judith Kleeman Professorial Chair.}  }

\begin{document}
\sloppy
\maketitle
\begin{abstract}
The focus of this work is \emph{hardness-preserving} transformations of  somewhat limited pseudorandom functions families (PRFs)  into ones with more versatile  characteristics. Consider the problem of \emph{domain extension} of pseudorandom functions: given a PRF that takes as input elements of some domain $\cU$, we would like to come up with a PRF over a larger domain. Can we do it with little work and without significantly impacting the security of the system? One approach is to first hash the larger domain into the smaller
one and then apply the original PRF. Such a reduction, however, is vulnerable to a ``birthday attack": after $\sqrt{\size{\cU}}$ queries to the resulting PRF,  a collision (\ie two distinct inputs having
the same hash value) is very likely to occur. As a consequence, the resulting PRF is \emph{insecure} against an attacker making this number of queries.

In this work we show how to go beyond the aforementioned birthday attack barrier by replacing the above simple hashing approach with a variant of  \textit{cuckoo hashing}, a hashing paradigm that resolves collisions in a table by using two hash functions and two tables, cleverly assigning each element to one
of the two tables. We use this approach to obtain: (i) a domain extension method that requires {\em just two calls} to the original PRF, can withstand as many
queries as the original domain size, and has a distinguishing probability that is exponentially small in the amount of non-cryptographic work; and (ii) a {\em security-preserving}
reduction from non-adaptive to adaptive PRFs.
\end{abstract}


\input{Introduction}
\input{Preliminaries}
\input{CuckooHashing}

\input{DomainExtension}
\input{NonAdaptiveToAdaptive}
\input{ADW}
\input{OpenQuestions}

\section{Acknowledgments}
We thank Eylon Yogev and the anonymous referees of TCC 2013 and Journal of Cryptology
for their helpful comments. The third author would like to thank his M.Sc
advisor Ran Raz for his support.

\bibliographystyle{myabbrvnat}
\bibliography{crypto}

\appendix
\input{NonAdaptiveFrameWork}
\input{PRGtoPRF}

\end{document}

%% file: macros.tex
\newcommand{\remove}[1]{}
\newcommand{\Draft}[1]{\ifnum\draft=1\texttt{ #1} \fi}
\newcommand{\Full}[1]{\ifnum\CR=0 #1 \fi}
\newcommand{\Ffootnote}[1]{\Full{\footnote{#1}}}

\ifnum\draft=1
    \newcommand{\authnote}[2]{{\bf [{\color{red} #1's Note:} {\color{blue} #2}]}}
\else
    \newcommand{\authnote}[2]{}
\fi

\newcommand{\sdotfill}{\textcolor[rgb]{0.8,0.8,0.8}{\dotfill}} 

\newenvironment{algorithm}{\begin{algo}}{\vspace{-\topsep}\end{algo}}

\newcommand{\aka} {also known as\ }
\newcommand{\resp}{resp.,\ }
\newcommand{\ie}  {i.e.,\ }
\newcommand{\eg}  {e.g.,\ }

\newcommand{\wrt} {with respect to\ }

\newcommand{\cf}{{cf.,\ }}

\newcommand{\set}[1]{\ens{#1}}
\newcommand{\paren}[1]{\left(#1\right)}

\newcommand{\floor}[1]{\left \lfloor#1 \right \rfloor}

\newcommand{\eqdef}{:=}

\newcommand{\N}{{\mathbb{N}}}

\newcommand{\F}{{\cal F}}

\newcommand{\zo}{\{0,1\}}
\newcommand{\zn}{{\zo^n}}

\newcommand{\xor}{\oplus}

\newcommand{\eps}{\varepsilon}

\newcommand{\la}{\gets}

\newcommand{\poly}{\operatorname{poly}}

\newcommand{\MathAlg}[1]{\mathsf{#1}}

\renewcommand{\cref}{\Cref}

\ifnum\CR=0
\theoremstyle{nicetheorem}

\newtheorem{theorem}{Theorem}[section]

\newaliascnt{lemma}{theorem}
\newtheorem{lemma}[lemma]{Lemma}
\aliascntresetthe{lemma}
\crefname{lemma}{Lemma}{Lemmas}

\newaliascnt{claim}{theorem}
\newtheorem{claim}[claim]{Claim}
\aliascntresetthe{claim}
\crefname{claim}{Claim}{Claims}

\newaliascnt{corollary}{theorem}
\newtheorem{corollary}[corollary]{Corollary}
\aliascntresetthe{corollary}
\crefname{corollary}{Corollary}{Corollaries}

\newaliascnt{construction}{theorem}

\aliascntresetthe{construction}
\crefname{construction}{Construction}{Constructions}

\newaliascnt{fact}{theorem}
\newtheorem{fact}[fact]{Fact}
\aliascntresetthe{fact}
\crefname{fact}{Fact}{Facts}

\newaliascnt{proposition}{theorem}
\newtheorem{proposition}[proposition]{Proposition}
\aliascntresetthe{proposition}
\crefname{proposition}{Proposition}{Propositions}

\newaliascnt{conjecture}{theorem}

\aliascntresetthe{conjecture}
\crefname{conjecture}{Conjecture}{Conjectures}

\newaliascnt{definition}{theorem}
\newtheorem{definition}[definition]{Definition}
\aliascntresetthe{definition}
\crefname{definition}{Definition}{Definitions}

\newaliascnt{remark}{theorem}
\newtheorem{remark}[remark]{Remark}
\aliascntresetthe{remark}
\crefname{remark}{Remark}{Remarks}
\fi

\newaliascnt{proto}{theorem}

\newtheorem{proto}[proto]{Protocol}

\aliascntresetthe{proto}
\crefname{proto}{protocol}{protocols}

\newaliascnt{algo}{theorem}
\newtheorem{algo}[algo]{Algorithm}
\aliascntresetthe{algo}
\crefname{algo}{algorithm}{algorithms}

\newaliascnt{expr}{theorem}
\newtheorem{expr}[expr]{Experiment}
\aliascntresetthe{expr}
\crefname{experiment}{experiment}{experiments}

\newaliascnt{notation}{theorem}

\aliascntresetthe{notation}
\crefname{notation}{Notation}{Notations}

\def\FullBox{$\Box$}
\def\qed{\ifmmode\qquad\FullBox\else{\unskip\nobreak\hfil
\penalty50\hskip1em\null\nobreak\hfil\FullBox
\parfillskip=0pt\finalhyphendemerits=0\endgraf}\fi}

\def\qedsketch{\ifmmode\Box\else{\unskip\nobreak\hfil
\penalty50\hskip1em\null\nobreak\hfil$\Box$
\parfillskip=0pt\finalhyphendemerits=0\endgraf}\fi}

\renewcommand{\Pr}{{\mathrm {Pr}}}

\newcommand{\concat}{\circ}

\newcommand{\ens}[1]{\{#1\}}
\newcommand{\size}[1]{\left|#1\right|}

\newcommand{\bsize}[1]{\bigr|#1\bigr|}

\newcommand{\Uni}{{\mathord{\mathcal{U}}}}

\newcommand{\prob}[1]{\mathsf{\textsc{#1}}}

\newcommand{\SD}{\prob{SD}}

\newcommand{\cM}{{\cal{M}}}
\newcommand{\cL}{{\cal{L}}}

\newcommand{\ppt}{{\sc pptm}\xspace}

\newcommand{\cH}{{\cal{H}}}

\newcommand{\cS}{\mathcal{S}}
\newcommand{\cU}{\mathcal{U}}
\newcommand{\cV}{\mathcal{V}}
\newcommand{\cD}{\mathcal{D}}
\newcommand{\cR}{\mathcal{R}}
\newcommand{\cT}{\mathcal{T}}

\newcommand{\cG}{\mathcal{G}}
\newcommand{\cF}{\mathcal{F}}

\newcommand{\cY}{\mathcal{Y}}

\newcommand{\hlen}{\ell}

\newcommand{\cu}{{\cal{U}}}
\newcommand{\cs}{{\cal{S}}}

\newcommand{\Dc}{{\MathAlg{D}}}
\newcommand{\Bad}{\ensuremath{\operatorname{BAD}}\xspace}

%% file: Introduction.tex
\section{Introduction}\label{sec:intro}
The focus of this work is \emph{hardness-preserving} transformations of somewhat
limited pseudorandom functions families (PRFs) into ones with more versatile
characteristics.  Examples of such somewhat limited families include those with
\emph{small} domain or those that can withstand only \emph{non-adaptive} (\aka
static) attacks, in which the attacker chooses its queries ahead of time, before
seeing any of the answers. In contrast, less limited families might have
\emph{large} domain or be secure against \emph{adaptive} (dynamic) attacks, in
which the attacker's queries might be chosen as a function of all previous
answers.


A common paradigm, first suggested by Levin \cite[\S 5.4]{Levin87}, for increasing the usability and security of a PRF, is to ``hash'' the inputs into a smaller domain \emph{before} applying
the PRF. This approach was originally suggested in order to achieve ``PRF domain extension" (using a short, \eg fixed,
input length PRF to get a variable-length PRF); more recently, it was used to transform non-adaptive PRFs into adaptive ones \cite{BermanH15}. Such reductions, however, are
vulnerable to the following ``birthday attack": after $\sqrt{\size{\cU}}$ queries to the resulting PRF, where $\cU$ is the hash function range, a collision
(\ie two distinct inputs having the same hash value) is very likely to occur. Such collisions are an obstacle to the indistinguishability of the
PRF, since in a random function we
either do not expect to see a collision at all (if the range is large enough) or expect to see fewer collisions. Hence, the resulting PRF is \emph{insecure}
against an attacker making this number of queries.

In this work we study variants of the above hashing approach to go beyond the
birthday attack barrier. At a high-level, our approach, which can be traced back
to \citet{Siegel04}, is based on
applying a dictionary data structure\footnote{In this context a dictionary is a data structure used
for maintaining a set of elements while supporting membership queries.}, in which the locations accessed in the search of an element are determined by its value and some fixed random string (\ie the same string is used for all elements), and not on values seen during the search. 
Now to do the conversion to domain extension we assign random values to all locations (by the underlying PRF). We think of the large domain as the universe from which the elements of the dictionary are taken.  The resulting value of the extended function at point $x$  will be some (simple) function of the values assigned to the
locations accessed during the search for $x$. For instance, one can view Levin's construction above as an instance of this framework, where the fixed random string describes a hash function from large domain to a smaller-size set $\cU$, and the PRF, whose domain is $\cU$, assigns random values for $\size{\cU}$
locations. The distinguishing probability of the resulting scheme is the distinguishing probability of the underlying PRF {\em plus} the probability of failure of the dictionary (which in  Levin's construction is determined by the ``birthday paradox").
The cost of the extension is related to the worst-case search time of the
dictionary (which in Levin's construction is a single invocation of the hash
function).

We focus on constructions based on \textit{cuckoo hashing}:  a hashing paradigm typically used for resolving hash collisions in a table by using two hash functions and two tables, assigning each element to one of the
two tables, and enabling lookup using only two queries (see \cref{sec:intro:Cuckoo}).
We use this paradigm to present a new PRF domain extension
method that requires \emph{just two calls} to the original PRF, can withstand as many queries as the original domain size, and has a distinguishing probability that is exponentially small in the amount of non-cryptographic work.
We also obtain a \emph{security-preserving reduction} from non-adaptive to adaptive PRFs, an
improvement upon the recent result of Berman and Haitner~\cite{BermanH15}.

Before stating our results, we discuss in detail pseudorandom functions and cuckoo hashing.

\subsection{Pseudorandom Functions}\label{sec:intro:PRF}
Pseudorandom function families (PRFs), introduced by \citet*{GoldreichGoMi86}, are function families that cannot be distinguished from a family of \emph{truly}
random functions by an efficient distinguisher given an oracle access to a random member of the family. PRFs have an extremely important role in
cryptography, allowing parties who share a common secret key to send secure messages, identify themselves, and authenticate messages
\cite{GoldreichGM84,Luby96}. They have many other applications as well, and can be used in just about any setting that requires a random function provided as a black-box
\cite{BellareG89, BlumEGKN94, ChorFNP00, Goldreich86a, LubyR88, Ostrovsky89}. Different PRF constructions, whose security is based on
different hardness assumptions, are known in the literature. The construction most relevant to this work is the one of \cite{GoldreichGoMi86}, hereafter the \GGM construction, which uses
a length-doubling pseudorandom generator (and thus can be based on the existence of one-way functions \cite{HastadImLeLu99}).

We use the following definitions: an efficiently computable  function family ensemble $\F= \set{\F_n}_{n\in \N}$ is a $(q,t,\eps)$-PRF, if (for large enough $n$) a $q(n)$-query oracle-aided algorithm (distinguisher) of running time $t(n)$, getting access to a random function from the family, distinguishes between
$\F_n$ and the family of all functions (with the same input/output domains), with probability at most $\eps(n)$. $\F$ is a  \emph{non-adaptive} $(q,t,\eps)$-PRF if it is only required to be secure against non-adaptive distinguishers (\ie ones that prepare all their queries in advance). Finally, $\F$ is a $t$-PRF if $q$ is only limited by $t$ and $\eps =1/t$.

We also make use of the information-theoretic analog of a $t$-PRF, known as a
$t$-wise independent family (see
\cref{def:pairwiseHash}).

\subsection{Cuckoo Hashing and Many-wise Independent Hash Function}\label{sec:intro:Cuckoo}
Cuckoo hashing, introduced by \citet{PaghR04}, is an efficient technique for
constructing dynamic dictionaries. Such data structures are used to maintain a
set of elements, while supporting membership queries as well as insertions and
deletions of elements. Cuckoo hashing maintains such a dynamic dictionary by
keeping two tables of size only slightly larger than the number of elements to
be inserted, and two hash functions mapping the elements into cells of those
tables. It then applies a clever algorithm for placing at most a single element
in each cell. Each membership query requires just two memory accesses (in the
worst case) and they are determined by the hash functions. Many variants of
cuckoo hashing have been proposed since its introduction, and extensive
literature has been devoted to its analysis (\cf
\cite{FotakisPSS05,DietzfelbingerW07,KirschMW09,FriezeMM11,ArbitmanNS10}).

\citet{PaghP08} used ideas in the spirit of cuckoo hashing to construct efficient many-wise independent hash functions. Let $\cH$, $\cG$ and $\F$ be function families from $\cD$ to $\cS$, from $\cD$ to $\cR$ and from $\cS$ to $\cR$ respectively, where $\cR$ is a group with operation $\oplus$. Define the function family $\PP(\cH,\cG,\cF)$ from $\cD$ to $\cR$ as
\begin{align*}
\PP(\cH,\cG,\cF) = (\cF\circ\cH) \oplus (\cF\circ\cH) \oplus \cG,
\end{align*}
where $\F_1 \circ \F_2$, for function families $\F_1$ and $\F_2$, is the function family whose members are the elements of $\F_1 \times \F_2$ and $(f_1,f_2)(x)$ is defined by $f_1(f_2(x))$ ($\F_1 \xor \F_2$ is analogously defined). In other words, given $f_1, f_2 \in \cF, h_1, h_2 \in \cH$ and $g \in \cG$, design a function $\PP_{f_1,f_2,h_1,h_2,g}(x) = f_1(h_1(x)) \oplus f_2(h_2(x)) \oplus g(x)$. \citet{PaghP08} showed that when the families $\cH$ and $\cG$ are of ``high enough" independence, that is, roughly $(c\cdot\log \size{\cS})$-wise independent, then the family $\PP(\cH,\cG,\Pi)$ is $O(\size{\cS}^{-c})$-indistinguishable from random by a $\size{\cS}$-query, \emph{non-adaptive} distinguisher, where $\Pi$ is the set of all functions from $\cS$ to $\cR$. Note that the security of the resulting family goes well beyond the birthday attack barrier: it is indistinguishable from random by an attacker making $\size{\cS}\gg \sqrt{\size{\cS}}$ queries.

\citet{AumullerDW14} (building on the work of \citet{DietzfelbingerW03}) strengthen the result of \cite{PaghP08} by using more sophisticated
hash functions $\cH$ and $\cG$ (rather than the $O(\log \size{\cS})$-wise
independent that \cite{PaghP08} used). Specifically, for a given $z\geq 0$,
\citet{AumullerDW14} constructed a function family $\ADW_z(\cH, \cG, \Pi)$
that is $O(\size{\cS}^{-(z+1)})$-indistinguishable
from random by a $\size{\cS}$-query, \emph{non-adaptive} distinguisher, where
$\Pi$ is the set of all functions from $\cS$ to $\cR$.\footnote{The $\ADW$'s function family is in fact more complicated than the above simplified description. See \cref{sec:InstantADW} for the formal definition.}
The idea to use more sophisticated hash
functions, in the sense that they require less combinatorial work, already
appeared in previous works, \eg the work of \citet[\S 5.4]{ArbitmanNS10}.

In \cref{sec:GeneralFrameworkMain} we take the above results a step further,
showing that they hold also for \emph{adaptive}
distinguishers.\Ffootnote{\label{fn:adaptiveVSnonadaptive}Note that in some
  cases an adaptive adversary is a more powerful distinguisher than a
  non-adaptive one. For example, when trying to distinguish between a truly
  random function and a random involution (permutations where the cycle length
  is at most 2), there exists an adaptive distinguisher that will succeed with
  very high probability by asking two queries while any non-adaptive
  distinguisher will fail with very high probability (see
  \cite{KaplanNR09,NaorR02}).} Our approach for this transformation has many
predecessors. For instance, the work of \citet{NaorR99}, and of
\citet{JetchevOS12}.  Furthermore, it turns out that by using the above function
family with a pseudorandom function $\F$ instead of a truly random function from
$\Pi$, namely the family $\PP(\cH,\cG,\F)$
(or $\ADW_z(\cH,\cG,\F)$), we get a pseudorandom function that is superior to
$\F$ (the actual properties of $\PP(\cH,\cG,\F)$ are determined by the
properties of $\F$ and the choice of $\cH$ and $\cG$). This understanding is the
main conceptual contribution of this paper, and the basis for the results
presented below.

We note that the works of \citet{PaghP08},  \citet{DietzfelbingerW03}, and \citet{AumullerDW14} have gone almost
unnoticed in the cryptography literature so far.\footnote{\citet{PaghP08} did notice this connection, and in particular mentioned the connection of their work to that of   \citet{BellareGK99}.} In this work we apply, in a black-box manner, the analysis of
\cite{PaghP08} and of  \cite{AumullerDW14} in cryptographic settings.

\subsection{Our Results}\label{sec:intro:OurResult}
We use a construction inspired by cuckoo hashing to improve upon two PRF reductions: PRF domain extension and non-adaptive to adaptive PRF.

\remove{For
example, a CBC-MAC using a function $f\colon\zn\mapsto\zn$ for an input $x=(x_1, x_2, \ldots, x_m) \in \zo^{mn}$ is the following tag:
\begin{align*}
f'(x) = f(f( \ldots f(f(x_1)\oplus x_2)\oplus\ldots\oplus x_{m-1}) \oplus x_m).
\end{align*}}

\subsubsection{PRF Domain Extension}\label{sec:Intro:DomainExtension}
PRF domain extensions use PRFs with ``small" domain size to construct PRFs with larger (or even unlimited) domain size. These extensions reduce the
cost of a single invocation of the PRF and increase its usability. Domain extension methods are typically measured by the security of the resulting PRFs, and by the
number of calls the resulting PRF makes to the underlying PRF.

Among the known domain extension techniques are
CBC-MAC and PMAC (a survey on their security can be found in
\cite{Nandi10}).  The number of calls made by these constructions to the
underlying (small domain) PRF can be as small as two. Assuming that the
underlying PRF is a random function over $\zn$, 
then  the resulting family is $(q,\unbounded,O(q^2/2^n))$-PRF (\ie the $\unbounded$ in the
second parameter means that the distinguisher's running time is unlimited).  A second technique is  the Feistel or Bene{\u{s}} transformations (\eg \cite{AiolloV96,Patarin04,Patarin08}, a survey of which can be found in \cite{Patarin10b}). The Bene\u{s} based construction makes $8$ calls to the underlying PRF and is $(q,\unbounded,O(q/2^n))$-PRF, whereas a $5$-round Feistel based construction (which makes $5$ calls to the underlying PRF) is
$(q,\unbounded,O(q/2^n))$-PRF.\footnote{This is by no means an exhaustive list of all
domain extension constructions that achieve beyond-birthday security.}

Our cuckoo hashing based function family (see below) is $(q,\infty,O(q/2^n))$-PRF and makes only two calls to the underlying PRF. Moreover, our construction can extend the domain size to any fixed $\poly(n)$ length, unlike the aforementioned constructions, which only double the domain size.\footnote{Of course, one can use these constructions to extend the domain to any $\poly(n)$ length by a recursive construction. This, however, will increase the number of calls to the underlying PRF by a polynomial factor.}

\begin{theorem}[informal]\label{thm:IntroPPDomainExtension}
Let $k\leq n$, let $\cH$ and $\cG$ be efficient $k$-wise independent function families mapping strings of length $\ell(n)$ to strings of length $n$, and let $\Pi$
be the family of all functions from $\zn$ to $\zn$. Then, the family $\PP(\cH,\cG,\Pi)$, mapping strings of length $\ell(n)$ to strings of length $n$, is a
$(q,\infty,q/2^{\Omega(k)})$-PRF, for $q\leq 2^{n-2}$.
\end{theorem}
For $k=\Theta(n)$, \cref{thm:IntroPPDomainExtension} yields a domain extension that is $(q,\unbounded,q/2^n)$-PRF, and makes only two
calls to the underlying PRF. Replacing in the above construction the function family $\PP(\cH,\cG,\Pi)$ with the family $\ADW_z(\cH,\cG,\Pi)$ yields a more versatile domain extension that offers a tradeoff between the number of calls to the PRF and the independence required for the hash functions. For details, see \cref{sec:dom_ext_adw}.

\paragraph{PRG to PRF reductions.}
\Cref{thm:IntroPPDomainExtension} can also be used to get a hardness-preserving
construction of PRFs from pseudorandom generators (PRG) in settings where there
is a non-trivial bound on the number of queries to the PRF. \citet{JainPT12}, who were the
first to propose this goal, noted that one can realize it using a domain
extension constructions. Thus, we apply \Cref{thm:IntroPPDomainExtension} to get
constructions of PRFs from PRGs which improves some of
the parameters of \citet{JainPT12}, but require longer keys. See \cref{sec:PRGImprove} for details.


\subsubsection{From Non-Adaptive to Adaptive PRF}\label{sec:Intro:NonAdaptive}
Adaptive PRFs can be constructed from non-adaptive ones using general techniques
such as using the PRG-based construction of \citet{GoldreichGoMi86} or the
\emph{synthesizer} based construction of \citet{NaorR99a}. These constructions,
however, make (roughly) $n$ calls to the underlying non-adaptive PRF (where $n$
is the input length).  Recently, \citet{BermanH15} showed how to perform this
security uplift at a much lower cost: the adaptive PRF makes only a
\emph{single} call to the non-adaptive PRF. Their construction, however, incurs
a significant degradation in the security: assuming the underlying function is a
non-adaptive $t$-PRF, then the resulting function is an (adaptive)
$O(t^{1/3})$-PRF. The reason for this degradation is the birthday attack we
mentioned earlier.

We present a reduction from non-adaptive to adaptive PRFs that \emph{preserves} the security of the non-adaptive PRF. The resulting adaptive PRF  makes only two calls to the underlying non-adaptive one.
\begin{theorem}[informal]\label{thm:IntroNonAdaptiveMain}
  Let $t$ be a polynomial-time computable integer function, let
  $\cH= \set{\cH_n \colon \zn \mapsto [4t(n)]_\zn}_{n\in\N}$ (where
  $[4t(n)]_\zn$ are the first $4t(n)$ elements of $\zn$) and
  $\cG=\set{\cG_n\colon\zn\mapsto\zn}_{n\in\N}$ be efficient $O(\log t(n))$-wise
  independent function families, and let $\F$ be a length-preserving
  non-adaptive $t$-PRF. Then, $\PP(\cH,\cG,\F)$ is a length-preserving
  $\left(t/4\right)$-PRF.
\end{theorem}

Our construction, as well as the one of \citet{BermanH15}, depends on the number
of queries made by the distinguisher. Thus, this does not give a single
transformation for all poly-time adversaries. The PRG-based and
synthesizer-based constructions are better in this sense as they are independent
of the distinguisher.\footnote{\citet{BermanH15} do give a single transformation
  for all poly-time adversaries with the additional cost of slightly increasing
  the number of calls to the underlying non-adaptive PRF. Their transformation
  can be carried to our setting and be applied with the $\PP$ function
  family. However, our construction is superior to \cite{BermanH15}'s only by
  reducing the security loss in a polynomial factor, a factor that is
  meaningless when considering all poly-time adversaries. Thus, we avoid stating
  our result in this all poly-time adversaries setting and refer to
  \cite{BermanH15} for such a result.}


Finally, as it was the case in the domain extension, replacing the function
family $\PP(\cH,\cG,\F)$ with the family $\ADW_z(\cH,\cG,\cF)$, yields a more
versatile non-adaptive to adaptive transformation that offers a trade-off
between the number of calls to the PRF and the requisite independence. See
\cref{sec:nonadp_to_adp_adw} for details.

\subsection{More Related Work}\label{sec:intro:RelatedWork}
\citet{BellareGK99} introduced a paradigm for using PRFs in the symmetric-key settings that, in retrospect, is similar to cuckoo hashing.  Assume that two parties, who share a secret function $f$ would like to use it for (shared-key) encryption. The `textbook' (stateless) solution calls for the sender to choose $r$ at random and send $(r,f(r)\oplus M)$ to the receiver, where $M$ is the message to be encrypted. This proposal breaks down if the sender chooses the same $r$ twice (in two different sessions with different messages). Thus, the scheme is
subject to the birthday attack and the length parameters should be chosen accordingly. This requires the underlying function to have a large domain. Instead, \cite{BellareGK99} suggested choosing $t>1$ values at random, and sending $(r_1, \ldots, r_t, f(r_1)\oplus \cdots \oplus
f(r_t) \oplus M)$. They were able to show much better security than the single
$r$ case. They also showed a similar result for message authentication. Our
domain extension results (see \cref{sec:DomainExtension}) improve upon the results of \cite{BellareGK99}.

The problem of transforming a scheme that is only resilient to non-adaptive attack into one that is resilient to adaptive attacks has received quite a lot of attention in the context of pseudorandom permutations (or block ciphers). \citet{MaurerP04} showed that, given a family of permutations that are information-theoretic secure against non-adaptive attacks, if two members of this family are independently composed, then the resulting permutation is also secure against adaptive attacks (see \cite{MaurerP04} for the exact formulation).
\citet{Pietrzak05} showed, however, that this is
not necessarily the case for permutations that are random-looking under a
computational assumption (see also \cite{Myers04,Pietrzak06}), reminding us that translating information-theoretic results to the computational realm is a tricky business.

\subsection*{Paper Organization}
Basic notations and formal definitions are given in
\cref{sec:prelim}. In \cref{sec:GeneralFrameworkMain} we formally define
the hashing paradigm of \citet{PaghP08}
and show how to extend their result to hold against adaptive adversaries. Our domain extension reduction based on \cite{PaghP08} is described in \cref{sec:DomainExtension}, and the improved
non-adaptive to adaptive reduction, also based on \cite{PaghP08}, is described in
\cref{sec:NonadaptivetoAdaptive}. In  \cref{sec:InstantADW} we present the more
advanced (and more complex) hashing paradigm of \citet{AumullerDW14}, and use it
to obtain a more versatile version of the above reductions. Some possible
directions for future research are discussed in \cref{sec:OpenQuestions}.


%% file: Preliminaries.tex
\section{Preliminaries}\label{sec:prelim}

\subsection{Notations}
All logarithms considered here are in base two. We use calligraphic letters to denote sets, uppercase for random variables, and lowercase for values. Let `$||$' denote string concatenation.  For an integer $t\in \N$, let $[t] = \set{1,\dots,t}$. For a set $\cS$, let $\cS^\ast$ be the power set of $\cS$ (\ie the set of all subsets of $\cS$). For a set $\cS$ and an integer $t\in \N$, let $\cS^{\leq t} = \set{\overline{s}\in\cS^\ast \colon \size{\overline{s}} \leq t \ \ \land \ \ \overline{s}[i] \neq \overline{s}[j] \ \ \forall i \neq j \in [\size{\overline{s}}]}$, where $\overline{s}[i]$ is the $i$th element of $\overline{s}$, and let $[t]_\cs$ be the first $t$ elements, in increasing lexicographic order, of $\cs$ (equal to $\cs$ in case $\size{\cs} \le t$). For sets $\cU$ and $\cV$, let $\Pi_{\cU\mapsto\cV}$ stands for the set of all functions from $\cu$ to $\cV$, and for integers $n$ and $\ell$, let $\Pi_{n,\ell} = \Pi_{\zn\mapsto\zo^\ell}$.

Let $\poly$ denote the set all polynomials, and let \ppt be abbreviation for probabilistic (strictly) polynomial-time Turing machine. For $s\in\N$ and $t,q\colon \N \mapsto \N$, we say that $\mD$ is a $t$-time $q$-query s-oracle-aided algorithm if, when invoked on input of length $n$, $\mD$ runs in time $t(n)$ and makes at most $q(n)$ queries to each of its $s$ oracles.
Given a random variable $X$, we write $X(x)$ to denote $\Pr[X=x]$, and write $x\gets X$ to indicate that $x$ is selected according to $X$. Similarly, given a finite
set $\cs$, we let $s\la \cs$ denote that $s$ is selected according to the uniform distribution on $\cs$. The \emph{statistical distance} of two distributions $P$
and $Q$ over a finite set $\Uni$, denoted as $\SD(P,Q)$, is defined as
$\max_{\cs\subseteq \Uni} \size{P(\cs)-Q(\cs)} = \frac{1}{2} \sum_{u\in \Uni}\size{P(u)- Q(u)}$.

\subsection{Pseudorandom Generators}
\begin{definition}[Pseudorandom Generators]
A polynomial-time function $G\colon \zn \mapsto \zo^{\ell(n)}$ is {\sf $(t,\eps)$-PRG}, if $\ell(n)>n$ for every $n\in\N$ ($G$ stretches the input), and
\begin{align*}
\size{\Pr_{x\la \zn}[\Dc(G(x))=1] - \Pr_{y\la \zo^{\ell(n)}}[\Dc(y)=1]} \leq \eps(n)
\end{align*}
for every algorithm (distinguisher) $\Dc$ of running time $t(n)$ and large
enough $n$. A $(t,1/t))$-PRG is called a $t$-PRG. If $\ell(n)=2n$, we
say that $G$ is \emph{length-doubling}.
\end{definition}

\subsection{Function Families}

\subsubsection{Operating on Function Families}
We consider two natural operations on function families.
\begin{definition}[composition of function families]\label{def:FunctionComp}
Let $\F^1\colon \cd^1 \mapsto \cR^1$ and $\F^2\colon \cd^2 \mapsto \cR^2$ be two function families with $\cR^1 \subseteq \cd^2$. The {\sf composition} of $\F^1$
with $\F^2$, denoted $\F^2 \circ \F^1$, is the function family $ \set{(f_2,f_1) \in \F^2 \times \F^1}$, where $(f_2,f_1)(x) \eqdef f_2(f_1(x))$.
\end{definition}

\begin{definition}[group operation of function families]\label{def:FunctionXOR}
Let $\F^1\colon \cd \mapsto \cR^1$ and $\F^2\colon \cd \mapsto \cR^2$ be two function families with $\cR^1,\cR^2 \subseteq \cR$, where $\cR$ is a group with operation $\oplus$. The {\sf group operation} of $\F^1$ with $\F^2$, denoted $\F^2 \bigoplus \F^1$, is the function family $\set{(f_2,f_1) \in \F^2 \times \F^1}$, where $(f_2,f_1)(x) \eqdef f_2(x) \oplus f_1(x)$.
\end{definition}

In all of our applications, the group $\cR$ from the above definition will simply be $\zn$ for some $n\in\N$, with XOR as the group operation.

\subsubsection{Function Family Ensembles}\label{sec:FunctionFamilies}
A function family ensemble is an infinite set of function families, whose elements (families) are typically indexed by the set of integers.  Let $\F= \set{\F_n\colon \cd_n \mapsto \cR_n}_{n \in \N}$ stands for an ensemble of function families, where each $f\in \F_n$ has domain $\cd_n$ and its range
is a subset of $\cR_n$. Such ensemble is \emph{length preserving}, if $\cd_n = \cR_n = \zn$ for every $n$. We naturally extend \cref{def:FunctionComp,def:FunctionXOR} to function family ensembles.

For a function family ensemble to be useful it should have an efficient sampling and evaluation algorithms.
\begin{definition}[efficient function family ensembles]\label{def:effFam}
A function family ensemble $\F=\set{\F_n\colon \cd_n \mapsto \cR_n}_{n\in\N}$ is {\sf efficient}, if the following hold:
\begin{description}
\item [Efficient sampling.] $\F$ is samplable in polynomial-time: there exists a \ppt that given $1^n$, outputs (the description of) a uniform element in $\F_n$.
\item [Efficient evaluation.] There exists a deterministic algorithm that given $x \in \cd_n$ and (a description of) $f \in \F_n$, runs in time $\poly(n,\size{x})$ and outputs $f(x)$.
\end{description}
\end{definition}

\subsubsection{Many-Wise Independent Hashing}\label{sec:pairwiseHash}
\begin{definition}[$k$-wise independent families]\label{def:pairwiseHash}
A function family $\cH= \set{h\colon \cd \mapsto \cR}$ is {\sf $k$-wise independent} (\wrt $\cd$ and $\cR$), if
\begin{align*}
\Pr_{h \la \cH}[h(x_1)= y_1 \land h(x_2)=y_2 \land \ldots \land h(x_k) = y_k] = \frac1{\size{\cR}^k},
\end{align*}
for every distinct $x_1,x_2,\ldots,x_k \in \cd$ and every $y_1,y_2,\ldots,y_k \in \cR$.
\end{definition}
For every $\hlen, k \in \poly$, the existence of efficient $k(n)$-wise independent family ensembles mapping strings of length $\hlen(n)$ to strings of length $n$ is
well known (\cite{CarterWe79,WegmanC81}). A simple and well known example of $k$-wise independent functions is the collection of all polynomials of degree $(k-1)$
over a finite field. This construction has small size, and each evaluation of a function at a given point requires $k$ operations in the field.
\begin{fact}\label{fact:k-indep}
For $\ell,n,k\in\N$, there exists a $k$-wise independent function family $\cH=\set{h\colon\zo^{\ell}\mapsto \zn}$, such that sampling a random element in $\cH$ requires $k\cdot\max\set{\ell,n}$ random bits, and evaluating a function from $\cH$ is done in time  $\poly(\ell,n,k)$.
\end{fact}

We mention that a $k$-wise independent families (as defined in
\cref{def:pairwiseHash}) look random for $k$-query distinguishers, \emph{both}
non-adaptive and adaptive ones. On the other hand, \textit{almost} $k$-wise
independent families\footnote{Formally, a function family
  $\cH=\set{h:\cD\mapsto\cR}$ is $(\eps,k)$-wise independent if for any
  $x_1,\dots,x_k\in\cD$ and for any $y_1,\dots,y_k\in\cR$ it holds that
  $\size{\Pr_{h\la\cH}[h(x_1)=y_1\wedge\dots\wedge h(x_k)=y_k]-\size{\cR}^{-k}}\leq
  \eps$. We call a family of functions an almost $k$-wise independent family,
  if it is $(\eps,k)$-wise independent for some small $\eps>0$.} are only
granted to be resistant against \emph{non-adaptive} distinguishers.\Ffootnote{See
  \Cref{fn:adaptiveVSnonadaptive} and references therein.} Yet, the
result presented in \cref{sec:GeneralFrameworkMain} yields that, in some cases,
the adaptive security of the latter families follows from their non-adaptive
security.

\subsection{Pseudorandom Functions}
\begin{definition}[Pseudorandom Functions]\label{def:PRF}
An efficient function family ensemble $\F=\set{\F_n\colon \zo^{m(n)} \mapsto \zo^{\ell(n)}}_{n \in \N}$ is an (adaptive)  {\sf $(q,t,\eps)$-PRF}, if for every $t$-time $q$-query oracle-aided algorithm (distinguisher) $\Dc$, it holds that
\begin{align*}
\size{\Pr_{f \leftarrow \F_n}[\Dc^f(1^n) = 1] - \Pr_{\pi \leftarrow \Pi_{m(n),\ell(n)}}[\Dc^{\pi}(1^n) = 1]} \leq \eps(n),
\end{align*}
for large enough $n$. If $q(n)$ is only bounded by $t(n)$ for every $n\in\N$, then $\F$ is called $(t,\eps)$-PRF. A $(t, 1/t)$-PRF is called a $t$-PRF.

If $\Dc$ is limited to be {\sf non-adaptive} (\ie it has to write all his oracle calls before making the first call), then $\F$ is called {\sf non-adaptive $(q,t,\eps)$-PRF}  (and we apply the above notational conventions also for this case).
\end{definition}

Some applications require the pseudorandom functions to be secure against distinguishers with access to many oracles (and not just a single oracle as
in \cref{def:PRF}).
\begin{definition}[pseudorandom functions secure against many-oracle distinguishers]
An efficient function family ensemble $\F=\set{\F_n\colon \zo^{m(n)} \mapsto
  \zo^{\ell(n)}}_{n \in \N}$ is an {\sf $s$-oracle $(q,t,\eps)$-PRF}, if for every $t$-time $q$-query $s$-oracle-aided algorithm (distinguisher) $\Dc$, it holds that
\begin{align*}
\size{\Pr_{\fv \leftarrow \F_n^s}[\Dc^{\fv}(1^n) = 1] - \Pr_{\piv \leftarrow \Pi_{m(n),\ell(n)}^s}[\Dc^{\piv}(1^n)
= 1]} \leq \eps(n),
\end{align*}
for large enough $n$.
\end{definition}

The following lemma shows that a standard (single oracle) PRF is also a
many-oracle one, with a loss that depends multiplicatively on the number of
oracles. This lemma has been used several times in the pase, for example, in
\cite[Lemma 3.3]{BellareCK96} and \cite[Theorem 1]{BilletEG10}. The lemma is
proven using a standard hybrid argument, see
\cite{BellareCK96,BilletEG10}.\footnote{\cite{BilletEG10} only states and proves
  the adaptive case, but the very same argument also yields the non-adaptive
  case.}

\begin{lemma}\label{lemma:TwoOracle}
Let  $\F = \set{\F_n \colon \zo^{m(n)} \mapsto \zo^{\ell(n)}}_{n\in\N}$ be a function family ensemble. Then for every $t$-time $q$-query {\sf $s$-oracle} adaptive [\resp non-adaptive] distinguisher $\Dc$, there exists a $\paren{t + s\cdot q \cdot e_\F }$-time $q$-query {\sf single-oracle} adaptive [\resp non-adaptive] distinguisher $\hcD$, where $e_\F$ stands for the evaluation time of $\F$,\footnote{That is, $\Dc(1^n)$ runs in time $t(n) + s\cdot q(n) \cdot e_\F(n)$. Moreover, we implicitly assume that the evaluation time of $\F$ is greater than the time needed to sample $\ell(n)$ random bits.} with
\begin{align*}
\lefteqn{\size{\Pr_{f \la \F_n}[\hcD^{f}(1^n) = 1] - \Pr_{\pi \la \Pi_{m(n),\ell(n)}}[\hcD^{\pi}(1^n) = 1]}}\\
&\geq \frac1s\cdot\size{\Pr_{\fv \leftarrow \F_n^s}[\Dc^{\fv}(1^n) = 1] - \Pr_{\piv \leftarrow \Pi_{m(n),\ell(n)}^s}[\Dc^{\piv}(1^n)
= 1]},
\end{align*}
for every $n\in \N$.
\end{lemma}

%% file: CuckooHashing.tex
\section{From Non-Adaptive to Adaptive Hashing}\label{sec:GeneralFrameworkMain}
In this section we show that non-adaptively secure function families with a certain combinatorial property are also adaptively secure, yielding that the function families of \citet{PaghP08} and \citet{AumullerDW14} are adaptively secure. In the next sections we take advantage of the latter implication to derive our hardness-preserving PRF reductions. We note that the above approach is not useful  for (arbitrary) non-adaptive PRFs, since a PRF might no posses the required  combinatorial property (indeed, not every non-adaptive PRF is an adaptive one).

To define the aforementioned combinatorial property, we use the notion of left-monotone sets.

\def\LeftMonotoneSets {Let $\cS$ and $\cT$ be sets. A set $\cM\subseteq\cS^\ast\times\cT$ is {\sf left-monotone}, if for every   $(\overline{s_1},t)\in \cM $ and  every $\overline{s_2} \in\cS^\ast$ that has
  $\overline{s_1}$ as a prefix, it holds that $(\overline{s_2},t) \in \cM$.}
\begin{definition}[left-monotone sets]\label{def:lhsmonotonesets}
  \LeftMonotoneSets
\end{definition}
Namely, a product set is left monotone, if it is monotone \wrt its
left-hand-side part, where all sequences having a prefix in a monotone set, are
also in the set. 

The main result of this section states roughly that for a function family $\F =
\set{f_{u,v}\colon \cD \mapsto \cR}_{(u,v)\in \cU \times\cV}$ and a
left-monotone set $\Bad\subseteq\cD^\ast \times \cU$, if for every set of
queries $\qb = (q_1, \ldots,q_{\size{\qb}})$ it holds that (1) for every $u\in
\cU$ such that $(\qb, u)\notin \Bad$, the outputs of $f$ on $q_1,
\ldots,q_{\size{\qb}}$ is uniform over $\cR^{\size{\qb}}$, and (2) the
probability over the choice of $u\la\cU$ that $(\qb, u)\in \Bad$ is small, then
$\F$ is adaptively secure.

\def\GeneralRandPiClose{Let $\cU$ and $\cV$ be
  non-empty sets,  let $\F = \F(\cU,\cV) = \set{f_{u,v}\colon \cD \mapsto
    \cR}_{(u,v)\in \cU \times\cV}$ be a function family and  let
  $\Bad\subseteq\cD^\ast \times \cU$ be left-monotone.  Let $t\in \N$, and
  assume that for every $\qb = (q_1, \ldots,q_{\size{\qb}})\in \cD^{\leq t}$ it
  holds that\footnote{Recall that for a set $\cS$  and an integer $t$, $\cS^{\leq t}$ denotes the set $\set{\overline{s}\in\cS^\ast
\colon \size{\overline{s}} \leq t \ \ \land \ \ \overline{s}[i] \neq \overline{s}[j] \ \ \forall  i \neq j \in [\size{\overline{s}}]}$.}
  \begin{enumerate}
  \item $\left(f(q_1),\dots, f(q_{\size{\qb}})\right)_{f\la \set{f_{u,v} \colon v\in\cV}}$ is uniform over  $\cR^{\size{\qb}}$, for every $u \in \cU$ with $(\qb,u)\not\in \Bad$, and
  \item $\Pr_{u\la\cU}[(\qb,u)\in\Bad] \leq \eps$.
  \end{enumerate}
Then,
  \begin{align*}
    \size{\Pr_{\substack{u\la\cU \\ v\la\cV}}[\mD^{f_{u,v}} = 1] - \Pr_{\pi \la
        \Pi}[\mD^{\pi}=1] } \leq \eps
  \end{align*}
for every $t$-query oracle-aided {\em adaptive} algorithm $\mD$, letting  $\Pi$ be the set of all functions from $\cD$ to $\cR$.}
\begin{lemma}\label{lemma:GeneralRandPiClose}
 \GeneralRandPiClose
\end{lemma}
Note that the above properties of the family $\F$ mean that  $\F$ is non-adaptively secure. The proof that $\F$ is also \emph{adaptively} secure, which critically uses the above structure of the set  $\Bad$, can be found in \cref{sec:GeneralFramework}.\footnote{\cref{lemma:GeneralRandPiClose} can be derived as  a special case of a result given in
\cite[Theorem~12]{JetchevOS12} (closing a gap in the proof appearing in \cite{Maurer02}). Yet, for the sake of completeness, we include an independent proof of this lemma here.}

\subsection{The \citet{PaghP08} Function Family}\label{sec:InstantPP}
We show that \cref{lemma:GeneralRandPiClose} can be applied to the function family of \citet{PaghP08}.
\begin{definition}[The \citet{PaghP08} function family]\label{def:PaghPagh}
  Let $\cH$ be a function family from $\cD$ to $\cU$, let $\cG$ be a function
  family from $\cD$ to $\cR$ and let $\cF$ be a function family from $\cS$ to
  $\cR$, with $\cU\subseteq\cS$ and $\cR$ being a group \wrt the operation $\oplus$.  The function family $\PP(\cH,\cG,\cF)$ from $\cD$ to $\cR$, is defined by
\begin{align*}
\PP(\cH,\cG,\cF) \eqdef (\cF\circ\cH) \oplus (\cF\circ\cH) \oplus \cG.
\end{align*}
For $h_1,h_2\in\cH$, let $\PP_{h_1,h_2}(\cG,\cF) \eqdef (\cF\circ h_1) \oplus (\cF\circ
h_2) \oplus \cG$.
\end{definition}
Graphically, this function family is given in \cref{fig:PP}. 
\begin{figure}[ht]
	\begin{center}\includegraphics[scale=0.4]{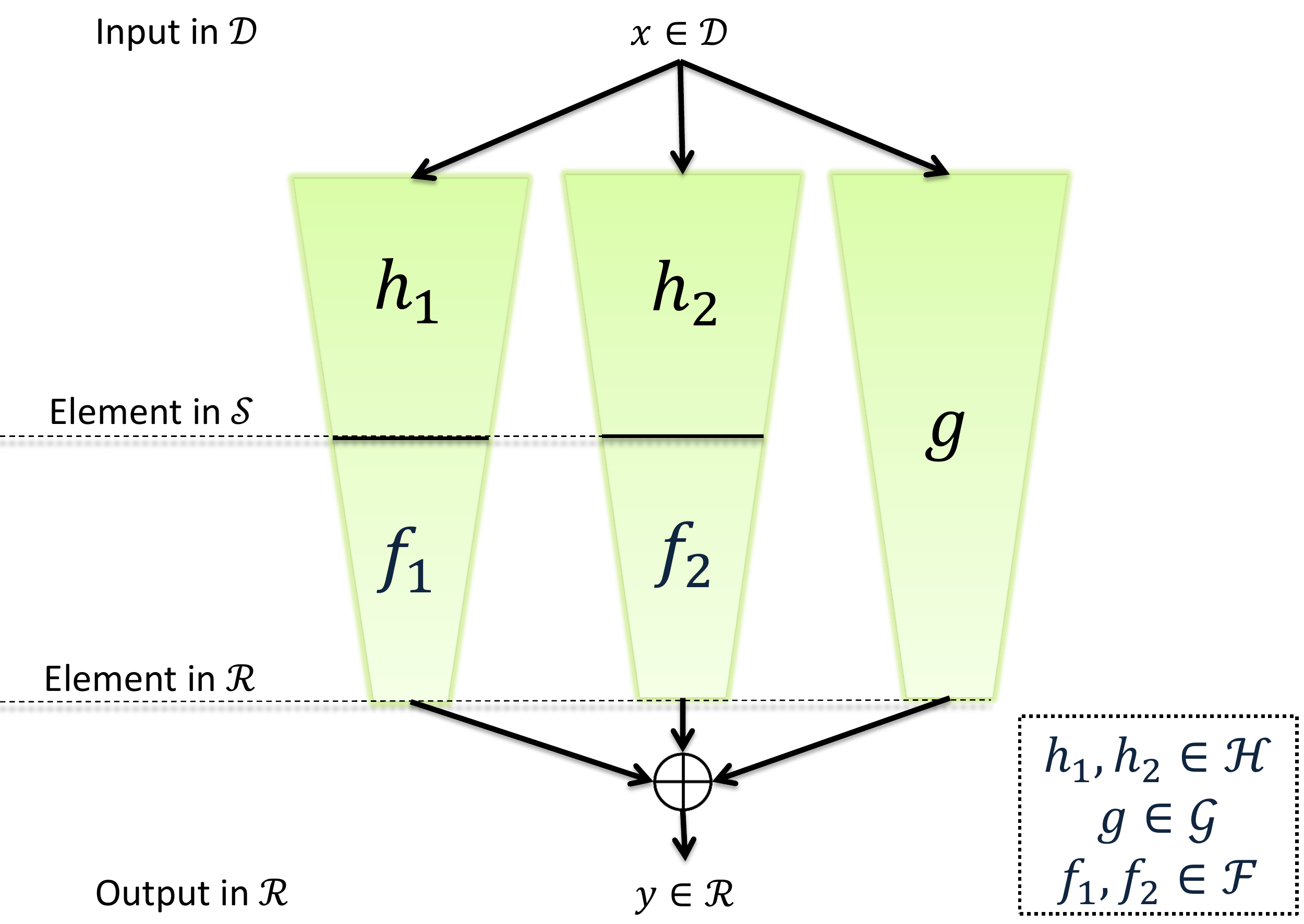}
          \caption{The function family $\PP(\cH,\cG,\cF)$. $\cH$ hashes down a
            domain $\cD$ to a domain $\cS$. Then $\cF$ maps $\cS$ to $\cR$. We
            do this twice and xor it with $\cG$, that hashes the domain $\cD$
            directly to $\cR$.}
	\label{fig:PP}\end{center}
\end{figure}
\citet{PaghP08} showed that when instantiated with the proper function families,
the above function family has the following properties:
\begin{theorem}[\cite{PaghP08}]\label{thm:PaghPagh}
    Let $t\in\N$, let $\cH = \set{h\colon \cD \mapsto \cU}$ and
    $\cG = \set{g\colon\cD\mapsto\cR}$ be function families with $\cR$ being a group \wrt the operation $\oplus$, and let $\Pi=\Pi_{\cS\mapsto\cR}$.

    If $\cU\subseteq\cS$ and $\size{\cU}\geq 4t$, then for every $k\in\N$ there exists a left-monotone set
    $\Bad \subseteq \cD^{\leq t}\times \cH^2$
    such that the following holds for every $\qb = (q_1,\ldots,q_{\size{\qb}})\in \cD^{\leq t}$:
\begin{enumerate}
\item\label{step:PPNotBad} Assuming that  $\cG$ is $k$-wise independent over the elements of $\qb$, then $\left(f(q_1),\dots, f(q_{\size{\qb}})\right)_{f\la \PP_{h_1,h_2}(\cG,\Pi)}$ is uniform over  $\cR^{\size{\qb}}$ for every $u \in \cU$ such that $(\qb,
    u)\not\in \Bad$.

  \item\label{step:PPProbBad} Assuming that  $\cH$ is $k$-wise independent over the elements of $\qb$, then
    $\Pr_{u \la \cH^2}[(\qb,u)\in\Bad] \leq t/2^{\Omega(k)}$.\footnote{The function family we  consider above (\ie $\PP$) is slightly different than the one given in
    \cite{PaghP08}. Their construction maps element $x\in\cD$ to $F_1[h_1(x)]\oplus F_2[h_2(x)] \oplus g(x)$, where $F_1$ and $F_2$ are
    uniformly chosen vectors from $\cR^t$, $h_1,h_2 \colon \cD \mapsto [t]$ are uniformly chosen from a function family $\cH$ and $g\colon \cD \mapsto \cR$ is chosen uniformly from a function family $\cG$. Yet, the correctness of \cref{thm:PaghPagh} follows in a straightforward manner from \cite{PaghP08} original proof (specifically from Lemma 3.3 and Lemma
    3.4).}
\end{enumerate}
\end{theorem}

\citet{PaghP08} concluded that for (the many) applications where the analysis is
applied \wrt a static set it is safe to use this family instead. However, as we can
see,  the function family  $\PP(\cH,\cG,\Pi)$ is not only close to being uniform in the eyes of a non-adaptive distinguisher, but also allows us to apply \cref{lemma:GeneralRandPiClose} to deduce its security in the  eyes of adaptive distinguishers. By plugging in \cref{thm:PaghPagh} into the general framework lemma
(\cref{lemma:GeneralRandPiClose}), we get the following result:
\begin{lemma}\label{lemma:PPUniStatClose}
Let $t\in\N$, let $\cH$, $\cG$ and $\Pi$ be as in \cref{thm:PaghPagh}, and let $\Dc$ be an {\sf adaptive}, $t$-query oracle-aided algorithm. Assuming that $\cH$ and $\cG$ are $k$-wise independent, then
\begin{align*}
\size{\Pr_{f\la\PP(\cH,\cG,\Pi)}[\Dc^f = 1] - \Pr_{\pi\la\Pi}[\Dc^\pi = 1]} \leq t/2^{\Omega(k)}.
\end{align*}
\end{lemma}
\begin{proof}
  Let $\cU=\cH\times\cH$ and $\cV=\Pi\times\Pi\times\cG$. For $(h_1,h_2)\in\cU$ and
  $(\pi_1,\pi_2,g)\in\cV$, let $F_{(h_1,h_2),(\pi_1,\pi_2,g)} = \pi_1\circ h_1
  \oplus \pi_2 \circ h_2 \oplus g$, and let $\F=\set{F_{u,v}\colon \cD \mapsto
    \cR}_{(u,v)\in \cU \times\cV}$. Finally, let \Bad be the set \Bad of
  \cref{thm:PaghPagh}. We prove the lemma showing that the above sets meet the
  requirements stated in \cref{lemma:GeneralRandPiClose}.

  \cref{step:PPNotBad} of \cref{thm:PaghPagh} and the assumed independence of $\cG$ and $\cH$, yield that the first requirement of
  \cref{lemma:GeneralRandPiClose} is satisfied.  \cref{step:PPProbBad} of \cref{thm:PaghPagh} yields that the second
  requirement of \cref{lemma:GeneralRandPiClose} is satisfied for $\eps= t/2^{\Omega(k)}$. Hence,  the proof of the lemma follows by \cref{lemma:GeneralRandPiClose}.
\end{proof}

\begin{remark}
For some of our applications, see \cref{sec:DomainExtension,sec:NonadaptivetoAdaptive,sec:InstantADW}, we need to apply \cref{lemma:PPUniStatClose} with  efficient $k$-wise independent function family ensembles mapping strings of length $n$ to the set $[t(n)]_\zn$, where $t$ is an efficiently computable function. It is easy to see (\cf \cite{BermanH15}) that such ensembles exist for any efficiently computable  $t$ that is a power of two. By considering $t'(n) = 2^{\floor{\log (t(n))}}$, we use these ensembles for our applications, while only causing factor of two loss in the resulting security.
\end{remark}

We use the above function family of \citet{PaghP08}  to extend the domain of
pseudorandom functions (see \cref{sec:DomainExtension}), and to transform a
non-adaptive pseudorandom function into an adaptive one (see \cref{sec:NonadaptivetoAdaptive}). In \cref{sec:InstantADW} we instantiate the above framework with the more advanced function family of \citet{AumullerDW14}, to get more versatile variants of the above applications.

%% file: DomainExtension.tex
\section{PRF Domain Extension}\label{sec:DomainExtension}
In this section we use the function family $\PP$ of \citet{PaghP08} (see
\cref{sec:GeneralFrameworkMain}) to extend a domain of a given PRF.

\begin{theorem}[Restating \cref{thm:IntroPPDomainExtension}]\label{thm:PPDomainExtension}
Let $\cH= \set{\cH_n \colon \zo^{d(n)} \mapsto \zo^{s(n)}}_{n\in\N}$ and $\cG = \set{\cG_n \colon \zo^{d(n)} \mapsto \zo^{r(n)}}_{n\in\N}$ be efficient $k(n)$-wise independent function family ensembles, and let $\F=\set{\F_n\colon\zo^{s(n)}\mapsto\zo^{r(n)}}_{n\in\N}$ be a $(q,t,\eps)$-PRF. Then,
$$\PP(\cH,\cG,\F) = \set{\PP(\cH_n,\cG_n,\F_n) \colon \zo^{d(n)}\mapsto\zo^{r(n)}}_{n\in\N}$$ is a
$(q,t-p\cdot q,2\eps + q/2^{\Omega(k)})$-PRF, where $p$ is a polynomial determined by the evaluation and sampling time of $\cH$, $\cG$ and $\F$ and $q(n)\leq 2^{s(n)-2}$ for every $n\in\N$.
\end{theorem}

To prove \cref{thm:PPDomainExtension} we first show that $\PP(\cH,\cG,\cF)$ is
computationally close to $\PP(\cH,\cG,\Pi)$, where $\cF$ is a PRF family and
$\Pi$ is the set of all functions. Then, we invoke \cref{lemma:PPUniStatClose}
that shows that $\PP(\cH,\cG,\Pi)$ is statistically close to the set of all
functions with appropriate domain and range.

\begin{proof}[Proof of \Cref{thm:PPDomainExtension}]
  Let $\cH, \cG$ and $\F$ be as in the statement.
\begin{claim}\label{lemma:DomainExtensionPP}
  For every $\paren{t - 2q\cdot (\eval_\cH + \eval_\cG + \eval_\cF)}$-time
  $q$-query distinguisher $\Dc$, where $\eval_\cH,\eval_\cG,\eval_\cF\colon
  \N\mapsto \N$ are the evaluation and sampling times of $\cH$, $\cG$ and $\cF$,
  respectively, and for all large enough $n$:
\begin{align*}
  \bsize{\Pr_{f \la \PP(\cH_n,\cG_n,\cF_n)}[\Dc^{f}(1^n) =
    1]-\Pr_{f \la \PP(\cH_n,\cG_n,\Pi_{s(n),r(n)})}[\Dc^f(1^n) = 1]} \leq
  2\eps(n).
\end{align*}
\end{claim}
\begin{proof}
  Assume that there is a distinguisher $\Dc$ as in the statement for which
  \begin{align*}
    \bsize{\Pr_{f \la \PP(\cH_n,\cG_n,\cF_n)}[\Dc^{f}(1^n) =
      1]-\Pr_{f \la \PP(\cH_n,\cG_n,\Pi_{s(n),r(n)})}[\Dc^f(1^n) = 1]} >
    2\eps(n).
  \end{align*}
  We show how it can be used to break the PRF. We define a $(t-2q(n)\cdot
  \eval_\F(n))$-time $q$-query two-oracle distinguisher $\hcT$:
  \begin{algorithm}[$\hcT$]\label{alg:PRGTwoOracle}~
    \begin{description}
    \item[Input:] $1^n$.
    \item[Oracle:]  functions $\phi_1,\phi_2$ from $\cD_n$ to $\cR_n$.
    \end{description}
    \begin{enumerate}
    \item Set $h_1,h_2\la\cH_n$, $g\la\cG_n$.
    \item Set $f=(\phi_1\circ h_1)\xor (\phi_2\circ h_2) \xor g$.
    \item Emulate $\Dc^f(1^n)$.
    \end{enumerate}
  \end{algorithm}
  Note that $\hcT(1^n)$ makes $q(n)$ queries to it oracle where each query
  consists of a call to $\phi_1$ and $\phi_2$ and an evaluation of $h_1,h_2$ and
  $g$, so it can be implemented to run in time as $\Dc$ plus at most $2q(n)\cdot
  (\eval_\cH(n) + \eval_\cG(n))$ which is exactly $(t-2q(n)\cdot
  \eval_\F(n))$. Observe that in case $\phi_1$ and $\phi_2$ are uniformly drawn
  from $\F_n$, then the emulation of $\Dc^f(1^n)$ done in
  $\hcT^{\phi_1,\phi_2}(1^n)$ is identical to a random execution of $\Dc^f(1^n)$
  with $f\la \PP(\cH_n,\cG_n,\cF_n)$. Similarly, in case $\phi_1$ and $\phi_2$
  are uniformly drawn from $\Pi_{s(n),r(n)}$, then the emulation is identical to
  a random execution of $\Dc^f(1^n)$ with $f\la
  \PP(\cH_n,\cG_n,\Pi_{s(n),r(n)})$. Thus,
\begin{align*}
  & \bsize{\Pr_{(f_1,f_2) \la \F_n\times\F_n}[\hcT^{f_1,f_2}(1^n)
      = 1] - \Pr_{(\pi_1,\pi_2) \la
        \Pi_{s(n),r(n)}\times\Pi_{s(n),r(n)}}[\hcT^{\pi_1,\pi_2}(1^n) = 1]}\\
  &= \bsize{\Pr_{f \la \PP(\cH_n,\cG_n,\cF_n)}[\Dc^{f}(1^n) = 1]-\Pr_{f \la
      \PP(\cH_n,\cG_n,\Pi_{s(n),r(n)})}[\Dc^f(1^n) = 1]}.\nonumber
\end{align*}
Hence, \cref{lemma:TwoOracle} yields that there exists a single-oracle distinguisher $\hcD$ that when invoked on input of length $n$ makes $q(n)$ queries to its oracle and runs in time $t(n) +  2q(n)\cdot(\eval_\cH(n) + \eval_\cG(n)) + 2q(n)\cdot\eval_\cF(n)\leq t(n) + 2q(n)\cdot(\eval_\cH(n) + \eval_\cG(n) + \eval_\cF(n))$, such that
\begin{align*}
  \lefteqn{\bsize{\Pr_{f \la \cF_n}[\hcD^{f}(1^n) = 1] - \Pr_{\pi \la
        \Pi_{s(n),r(n)}}[\hcD^\pi(1^n) = 1]}} \\
  &\geq \frac12 \cdot\bsize{\Pr_{f \la \PP(\cH_n,\cG_n,\cF_n)}[\Dc^{f}(1^n) =
    1]-\Pr_{f \la \PP(\cH_n,\cG_n,\Pi_{s(n),r(n)})}[\Dc^f(1^n) = 1]} > \eps(n),
\end{align*}
which is a contradiction.
\end{proof}

Let $\Dc$ be a $q$-query $(t-p\cdot q)$-time distinguisher, for
$p\eqdef 2(\eval_\cH + \eval_\cG + \eval_\cF)$. Since $\cH$ and $\cG$ are
$k$-wise independent function families, we invoke \cref{lemma:PPUniStatClose}
(instantiated with the domain of bit strings by setting $t=q$, $\cD=\zo^{d}$,
$\cU=\cS=\zo^{s}$ and $\cR=\zo^{r}$; note that $\size{U} = 2^s \geq 4q =4t$) and
using \cref{lemma:DomainExtensionPP} and the triangle inequality to get that
\begin{align*}
  \size{\Pr_{f \la \PP(\cH_n,\cG_n,\cF_n)}[\Dc^{f}(1^n) = 1] - \Pr_{\pi\la\Pi_{d(n),r(n)}}[\Dc^\pi(1^n)=1]} \leq 2\eps(n) + q(n)/2^{\Omega(k(n))}
\end{align*}
for large enough $n$.
\end{proof}

Note that in order for \cref{thm:PPDomainExtension} to be useful, we have to set $k(n) = \Omega(\log q(n))$. In \cref{sec:InstantADW} we show how to achieve domain extension using functions with less independence, but with the cost of additional calls to the PRF.



%% file: NonAdaptiveToAdaptive.tex
\section{From Non-Adaptive to Adaptive PRF}\label{sec:NonadaptivetoAdaptive}
In this section we use the function family $\PP$ of \citet{PaghP08} (see
\cref{sec:GeneralFrameworkMain}), to transform an non-adaptive PRF into an
adaptive one in a security preserving manner. To ease notations, we assume that
the given non-adaptive PRF is length preserving.

\begin{theorem}\label{thm:PPNonAdaptiveMain}
  Let $q$ be a polynomial-time computable integer function with $q(n)\leq 2^{n-2}$ for every $n\in\N$, let $\cH= \set{\cH_n
    \colon \zn \mapsto [4q(n)]_\zn}_{n\in\N}$ and $\cG = \set{\cG_n \colon \zn
    \mapsto \zn}_{n\in\N}$ be efficient $(c\cdot \log q)$-wise independent
  function family ensembles, where $c>0$ is universal.

If  $\F = \set{\F_n
    \colon \zn \mapsto \zn}_{n\in\N}$ is a non-adaptive
  $(4q,p\cdot t,\eps)$-PRF for some $p\in\poly$ determined by the
  evaluation time of $q,\cH,\cG$ and $\F$, then $\PP(\cH,\cG,\cF)$ is an
  \emph{adaptive} $(q,t, 2\eps+1/q)$-PRF.
\end{theorem}

\cref{thm:IntroNonAdaptiveMain} is a special case of
\cref{thm:PPNonAdaptiveMain}, starting with a $(p\cdot q)$-PRF, where
$p\in\poly$ is determined by $q$ and the evaluation time of $\cH,\cG$ and $\F$
(assuming that $p(n) \ge 4$), and ending with $\PP(\cH,\cG,\cF)$ which is an
adaptive $q$-PRF.

To prove \cref{thm:PPNonAdaptiveMain} we begin by showing that if
$\PP(\cH,\cG,\cF)$ is distinguishable from $\PP(\cH,\cG,\Pi)$ by an adaptive
distinguisher, where $\cF$ is \emph{non-adaptive} pseudorandom function and
$\Pi$ is a truly random function with the same domain and range as $\cF$, then
there is a non-adaptive distingusiher that can succeed in the same task almost
as well. Then, we invoke \cref{lemma:PPUniStatClose} to get the theorem.
\begin{proof}[Proof of \cref{thm:PPNonAdaptiveMain}]
  Let $\cH, \cG$ and $\F$ be as in the statement and let
  $p=\eval_q +8q(\eval_\cH + \eval_\cG + \eval_\cF)$, where $\eval_q$ is the
  evaluation time of $q$, and $\eval_\cH$, $\eval_\cG$ and $\eval_\cF$ are the
  sampling and evaluation time of $\cH$, $\cG$ and $\cF$, respectively.
  \begin{claim}\label{lemma:NAPP}
    For every $t$-time $q$-query oracle-aided adaptive distinguisher $\Dc$,
    there exists a $\paren{p\cdot t}$-time $(4q)$-query, non-adaptive, oracle-aided distinguisher
    $\hcD$ such that for every $n\in \N$ and $q(n)\leq
    2^{n-2}$:
    \begin{align*}
      \lefteqn{\bsize{\Pr_{f \la \cF_n}[\hcD^{f}(1^n) = 1] - \Pr_{\pi \la
            \Pi_{n}}[\hcD^\pi(1^n) = 1]}} \\
      &\geq \frac12 \cdot\bsize{\Pr_{f \la \PP(\cH_n,\cG_n,\cF_n)}[\Dc^{f}(1^n) =
        1]-\Pr_{f \la \PP(\cH_n,\cG_n,\Pi_{n})}[\Dc^f(1^n) = 1]}.
    \end{align*}
  \end{claim}

  \begin{proof}
    The proof follows along similar lines to the proof of \cite[Lemma 3.3]{BermanH15}.
    Let $\hcT$ be the following two-oracle distinguisher:
    \begin{algorithm}[$\hcT$]\label{alg:TwoOracleDis}~
      \begin{description}
      \item[Input:] $1^n$.
      \item[Oracles:] Functions $\phi_1$ and $\phi_2$ from $\zn$ to $\zn$.
      \end{description}
      \begin{enumerate}
      \item \label{Step:UniFirst} Compute $\phi_1(x)$ and $\phi_2(x)$ for every $x \in
        [4q(n)]_{\zn}$.

      \item Set $f=(\phi_1\circ h_1)\xor (\phi_2\circ h_2) \xor g$, where $h_1,h_2\la\cH_n$ and $g\la\cG_n$.

      \item Emulate $\Dc^f(1^n)$: answer a query $x$ to $\phi_1$ and $\phi_2$ made by
        $\Dc$ with $f(x)$, using the information obtained in Step
        \ref{Step:UniFirst}.
      \end{enumerate}
    \end{algorithm}
    Note that $\hcT(1^n)$ makes $4q(n)$ \emph{non-adaptive} queries to $\phi_1$ and $\phi_2$, and it can be implemented to run in time $\eval_q(n) + 8q(n) + t(n) + q(n)\cdot(\eval_\cH(n) + \eval_\cG(n))$. Observe that the definition of $\PP$ guarantees that $\phi_1$ and $\phi_2$ are only queried on the first $4q(n)$ elements of their domain. Hence, in case
    $\phi_1$ and $\phi_2$ are uniformly drawn from $\F_n$, then the emulation of $\Dc^f(1^n)$ done in $\hcT^{\phi_1,\phi_2}$ is identical to a random execution of $\Dc^f(1^n)$ with $f\la \PP(\cH_n,\cG_n,\cF_n)$. Similarly, in case $\phi_1$ and $\phi_2$ are uniformly drawn from $\Pi_n$, then the emulation is identical to a random execution of $\Dc^f(1^n)$ with $f\la \PP(\cH_n,\cG_n,\Pi_n)$. Thus,
    \begin{align*}
      \lefteqn{\size{\Pr_{(f_1,f_2) \leftarrow \F_n\times\F_n}[\hcT^{f_1,f_2}(1^n) =
          1] - \Pr_{(\pi_1,\pi_2) \leftarrow
            \Pi_n\times\Pi_n}[\hcT^{\pi_1,\pi_2}(1^n) = 1]}}\\
      &= \size{\Pr_{f \la \PP(\cH_n,\cG_n,\cF_n)}[\Dc^{f}(1^n) = 1]-\Pr_{f \la
          \PP(\cH_n,\cG_n,\Pi_n)}[\Dc^f(1^n) = 1]}.
    \end{align*}
    Hence, \cref{lemma:TwoOracle} yields that there exists a non-adaptive,
    single-oracle distinguisher $\hcD$ that when invoked on input of length $n$
    makes $4q(n)$ queries and runs in time
    $\eval_q(n) + 8q(n) + t(n) + q(n)\cdot(\eval_\cH(n) + \eval_\cG(n)) +
    2q(n)\cdot\eval_\F(n) \leq p(n)\cdot t(n)$, such that
    \begin{align*}
      \lefteqn{\size{\Pr_{f \la \F_n}[\hcD^{f}(1^n) = 1] - \Pr_{\pi \la
            \Pi_n}[\hcD^{\pi}(1^n) = 1]}}\\
      &\geq \frac12\cdot\size{\Pr_{f \la \PP(\cH_n,\cG_n,\cF_n)}[\Dc^{f}(1^n) = 1]-\Pr_{f \la
          \PP(\cH_n,\cG_n,\Pi_n)}[\Dc^f(1^n) = 1]},
    \end{align*}
    for every $n\in\N$.
  \end{proof}

  Since $\F$ is a non-adaptive $(4q,p\cdot t,\eps)$-PRF (for a large enough
  polynomial $p$ as defined above), \cref{lemma:NAPP} implies that for every
  adaptive $t$-time $q$-query
  oracle-aided algorithm $\Dc$, it holds that
\begin{align*}
\size{\Pr_{f \la \PP(\cH_n,\cG_n,\cF_n)}[\Dc^{f}(1^n) = 1]-\Pr_{f \la
      \PP(\cH_n,\cG_n,\Pi_n)}[\Dc^f(1^n) = 1]} \leq 2\eps(n),
\end{align*}
for large enough $n$.  Since $\cH$ and $\cG$ are $k(n)$-wise independent, we can
use \cref{lemma:PPUniStatClose} (instantiated for our purpose by setting $t=q$,
$\cD=\zo^{d}$, $\cS=\zn$, $\cU=[4q]_{\zn}$ and $\cR=\zn$; note that $\size{U}=4q=4t$) and the triangle
inequality to get that
\begin{align*}
  \size{\Pr_{f \la \PP(\cH_n,\cG_n,\cF_n)}[\Dc^{f}(1^n) = 1]-\Pr_{\pi \la
      \Pi_n}[\Dc^\pi(1^n) = 1]} \leq 2\eps(n) + q(n)/2^{\Omega(k(n))}
\end{align*}
for large enough $n\in\N$. Setting $k(n)=\Theta(\log q(n))$ finishes the proof.
\end{proof}


%% file: ADW.tex
\newcommand{\const}{\mathsf{const}}
\section{Hardness-Preserving Reductions via Advanced Cuckoo Hashing}\label{sec:InstantADW}
In this section we apply the reductions given in
\cref{sec:DomainExtension,sec:NonadaptivetoAdaptive} using the
 function family of \citet{PaghP08},  with the function family  of
\citet{AumullerDW14}, to get a more versatile reduction (for  comparison see  \cref{sec:pp_adw_comp}). Roughly speaking, the function family of \citet{AumullerDW14} requires less
combinatorial work (\ie smaller independence) than the \citet{PaghP08}
family. On the other hand, the function family of \citet{AumullerDW14}  requires more ``randomness'' (\ie has a longer description) and is harder to describe. In \cref{subsec:InstantADW} we formally define the hash function family of
\citet{AumullerDW14}, state their (non-adaptive) result, and apply
\cref{lemma:GeneralRandPiClose} to get an adaptive variant of this result. In
\cref{sec:dom_ext_adw} we use the function family of \citet{AumullerDW14} to obtain a  PRF domain extension, and in
\cref{sec:nonadp_to_adp_adw} we use it to get a non-adaptive to adaptive
transformation of PRFs.

\subsection{The \citet{AumullerDW14} Function Family}\label{subsec:InstantADW}
The function family of \citet{AumullerDW14} (building upon
\citet{DietzfelbingerW03}) follows the same basic outline as the \citet{PaghP08}
function family, but uses more complex hash functions.  Recall that the members
of the \citet{PaghP08} function family $\PP(\cH,\cG,\cF)$ are of the form
$(f_1\circ h_1) \oplus (f_2\circ h_2) \oplus g$, for $f_1,f_2\in \cF$,
$h_1,h_1,\in \cH$ and $g\in \cG$.  In the function family
$\ADW(\cH, \cL, \cG, \cF, \cM, \cY)$ described below, the role of $h_1,h_2\in\cH$ is taken by
some variant of \textit{tabulation hashing} (and not simply from a relatively
high $k$-wise independent family as in \cite{PaghP08}). At the
heart of these functions lies a function $\adww{h,\gv,\mv}:\cD\mapsto\cS$ of the
form:
$$\adww{h, \gv, \mv}(x) \eqdef h(x) \oplus \bigoplus_{1\leq i \leq \cc}\m_i(\g_i(x)),$$
for $h \colon \cD \mapsto
\cS$, $\gv = (\g_1,\cdots, \g_\cc)$ and $\mv = (\m_1,\cdots,\m_\cc)$, where
 $\g_i\colon \cD \mapsto \cU$  and
$\m_i\colon \cU \mapsto \cS$. Jumping ahead, the $\m_i$'s will be chosen to be random functions (or
pseudorandom functions)
and the $\g_i$'s and $h$ will be chosen from a relatively low independence family. The \citet{AumullerDW14} construction uses several functions of the above form that, unlike  \cite{PaghP08}, are chosen in a \emph{correlated}
manner (in particular,  sharing the \emph{same} function vector
$\gv$). The precise definition is:

\begin{definition}[The \citet{AumullerDW14} function family]\label{def:ADWHashFunction}
For $z\in\N$ let
\begin{enumerate}
  \item $\cD,\cU$ be sets and $\cS,\cR$ be commutative groups defined
    \wrt an operation $\xor_{\cS}$ and $\oplus_{\cR}$, respectively (we will omit the subscript when it is clear);
  \item function families $\cH = \set{h\colon \cD\mapsto\cS}$, $\cL=\set{\ell\colon
      \cD\mapsto \cR}$, $\cF = \set{f\colon \cS \mapsto \cR}$, $\cG=\set{g\colon
      \cD\mapsto \cU}$, $\cM = \set{m\colon \cU \mapsto \cS }$ and $\cY =
    \set{y\colon \cU \mapsto \cR }$;
  \item functions $h_1, h_2\in \cH$, $\ell\in \cL$,
    $f_1, f_2\in \cF$;
  \item function vector $\gv = (\g_1,\cdots, \g_z)$, where $\g_i\in\cG$ for every $1\leq i \leq z$;
  \item function vectors $\mv^1 = (m^1_1,\ldots,m^1_z),\mv^2 = (m^2_1,\ldots,m^2_z)$, where $m^j_i\in\cM$ for each $j\in\set{1,2}$ and $1\leq i \leq z$; and
  \item function vector $\yv=(y_1,\ldots,y_z)$, where $y_i\in \cY$ for $1\leq i \leq z$.
\end{enumerate}
  Define $\adw_{\mv^1,\mv^2,\yv,h_1,h_2,\ell,\gv,f_1,f_2} \colon \cD \mapsto\cR$ by
 \begin{align*}
    \adw_{\mv^1, \mv^2,y,h_1,h_2,\ell,\gv,f_1,f_2} \eqdef (f_1\circ\adww{h_1, \gv,\mv^1}) \xor_\cR (f_2\circ\adww{h_2, \gv,\mv^2})\xor_\cR \adww{\ell, \gv,\yv},
  \end{align*}
  where $\adww{h, \gv,\mv}(x) \eqdef h(x) \oplus \bigoplus_{1\leq i \leq
    \cc}\m_i(g_i(x))$.\footnote{Note that $\adww{h_1, \gv,\mv^1},\adww{h_2,
    \gv,\mv^2}\colon\cD\mapsto\cS$ uses $\xor_\cS$ and $\adww{\ell,
    \gv,\yv}\colon\cD\mapsto\cR$ uses $\xor_\cR$.}

 For $\mv^1,\mv^2\in\cM^\cc$, $h_1,h_2\in\cH$ and $\gv\in\cG^\cc $, function
 family $\cL,\cF$ and $\cY$ as above, let
 \begin{align*}
   \ADW_{\adwparam,(\mv^1,\mv^2,h_1,h_2,\gv)}(\cL,\cF,\cY) \eqdef \set{\adw_{\mv^1,
         \mv^2,\yv,h_1,h_2,\ell,\gv,f_1,f_2} \colon \yv\in\cY^\cc,\ell\in
       \cL,f_1,f_2\in\cF}.
   \end{align*}
   Finally, let
 \begin{align*}
   \ADW_{\cc}(\cH, \cL, \cG, \cF, \cM, \cY) \eqdef &\set{\adw_{\mv^1,
       \mv^2,\yv,h_1,h_2,\ell,\gv,f_1,f_2} \colon \\&
     \mv^1,\mv^2\in\cM^\cc,\yv\in\cY^\cc,h_1,h_2\in\cH, \nonumber\\&\ell\in
     \cL,\gv\in\cG^\cc,f_1,f_2\in\cF}.
\end{align*}

\end{definition}
Graphically, this function family is described in \cref{fig:ADW}.
\begin{figure}[ht]
	\begin{center}\includegraphics[scale=0.55]{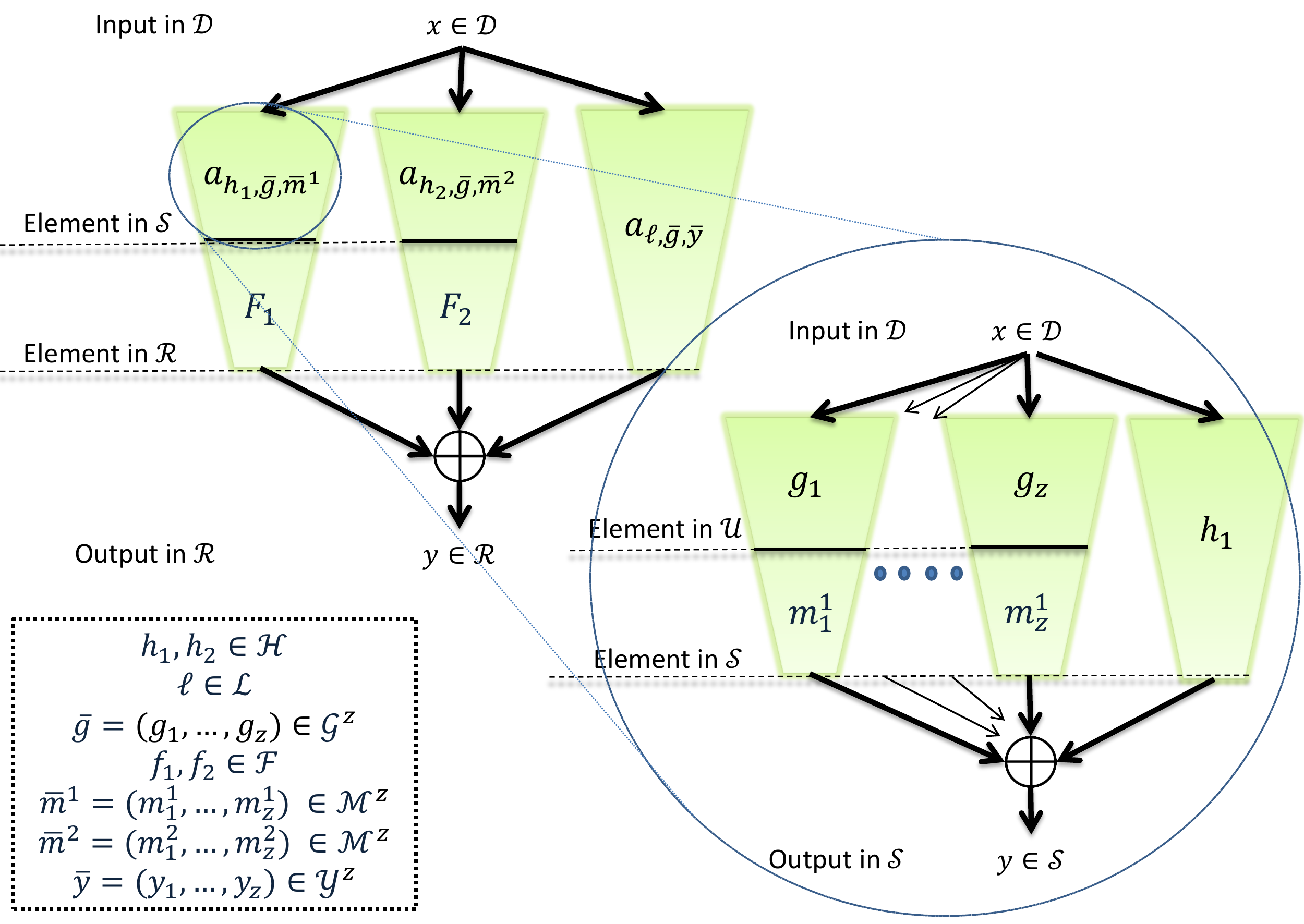}
          \caption{The function family $\ADW_{\cc}(\cH, \cL, \cG, \cF, \cM,
            \cY)$ is in the top left corner. On the bottom right corner the
            function $\adww{h_1, \gv, \mv^1}$ is depicted. The function
            $\adww{h_2, \gv, \mv^2}$ is similar. In $\adww{\ell, \gv, \yv}$ the
            range of the functions $\ell$ and $y_1\dots,y_z$ is $\cR$ (rather
            than $\cS$).}
	\label{fig:ADW}\end{center}
\end{figure}




\citet{AumullerDW14} proved the following result \wrt the above function family.
\begin{theorem}[\cite{AumullerDW14}]\label{thm:ADWOurSetting}
Let $t,k,\cc\in \N$ and let $\cD,\cR,\cS,\cU$ be
commutative groups defined \wrt an operation $\xor$ such that $\size{\cU}\in[t]$
and $\size{\cS}\geq
4t$.\footnote{The actual setting in \cite{AumullerDW14} is more
general. Specifically, an additional parameter $\eps>0$ is used to set the size of $\cS$ as $(1+\eps)t$. For simplicity of presentation, comparison
with the statement of \cref{thm:PaghPagh}, and since we use this theorem only
when $m$ is an integer, we set $\eps=3$.} Assume that $\cc \cdot k \in O(\log
t)$. Let $\cH = \set{h\colon \cD \mapsto
\cS}$, $\cL=\set{\ell\colon \cD\mapsto \cR}$ and $\cG=\set{g\colon
      \cD\mapsto \cU}$
be $2k$-wise independent hash families.

Then, there exist a universal constant  $\const>0$ and a left-monotone set $\Bad \subseteq \cD^{\leq t}\times \paren{\paren{\paren{\Pi_{\cU\mapsto\cS}}^\cc}^2\times\cH^2\times\cG^\cc}$,
such that the following holds for every $\qb=(q_1, \ldots, q_{\size{\qb}})\in \cD^{\leq t}$:

\begin{enumerate}

\item\label{step:ADWNotBad}
  $\left(f(q_1),\dots, f(q_{\size{\qb}})\right)_{f\la
    \ADW_{z,u}(\cL,\Pi_{\cS\mapsto\cR},\Pi_{\cU\mapsto\cR})}$ is uniform over
  $\cR^{\size{\qb}}$ for every
  $u =(\mv^1,\mv^2,h_1,h_2,\gv) \in \paren{\paren{\Pi_{\cU\mapsto\cS}}^\cc}^2\times\cH^2\times\cG^\cc$ such
  that $(\qb, u)\not\in \Bad$, and

\item\label{step:ADWProbBad} $\Pr_{u
    \la \paren{\paren{\Pi_{\cU\mapsto\cS}}^\cc}^2\times\cH^2\times\cG^\cc}[(\qb,u)\in\Bad]
  = \const\cdot t/\size{\cU}^{\cc \cdot k/2}$.
\end{enumerate}

\end{theorem}

\begin{remark}
  The construction from \cref{def:ADWHashFunction} and \cref{thm:ADWOurSetting}
  are taken from \cite[Section 6]{AumullerDW12} which is the conference version
  of \cite{AumullerDW14}. In \cite{AumullerDW12} the authors only considered the
  case where $\cc \cdot k$ is constant, independent of $t$. The proof for the
  case where $\cc \cdot k \in O(\log t)$ follows from the proof of
  \cite[Theorem 2]{AumullerDW14}.
\end{remark}

Applying \cref{lemma:GeneralRandPiClose},  the above yields that, for the right choice of parameters,  the function family
$\ADW_{\cc}(\cH, \cL, \cF, \cG, \cM, \cY)$ is not only close to being uniform in the eyes of a \emph{non-adaptive} distinguisher, but also in the  eyes of an \emph{adaptive} one. Specifically, combining \cref{lemma:GeneralRandPiClose,thm:ADWOurSetting} yields the following result.

\begin{lemma}\label{lemma:ADWUniStatCloseGen}
  Let $t,k,z,\const,\cD,\cR,\cS,\cU,\cH,\cL,\cG$ be as in \cref{thm:ADWOurSetting} and let $\Dc$ be an
  {\sf adaptive}, $t$-query oracle-aided algorithm. Then,
\begin{align*}
  \size{\Pr_{f\la\ADW_\adwparam(\cH,\cL,\cG,\Pi_{\cS\mapsto\cR},\Pi_{\cU\mapsto\cS},\Pi_{\cU\to\cR})}[\Dc^f
    = 1] - \Pr_{\pi\la\Pi}[\Dc^\pi = 1]} \le \const\cdot t/\size{\cU}^{\cc\cdot
    k/2}.
\end{align*}
\end{lemma}
\begin{proof}
  Let $\cU'=\paren{\paren{\Pi_{\cU\mapsto\cS}}^\cc}^2\times\cH^2\times\cG^\cc$, $\cV=\paren{\Pi_{\cU\to\cR}}^\cc\times(\Pi_{\cS\to\cR})^2\times\cL$. For
  $(\mv^1,\mv^2,h_1,h_2,\gv)\in\cU'$ and $(\yv,\pi_1,\pi_2,\ell)\in\cV$, let
  \begin{align*}
    F_{(\mv^1,\mv^2,h_1,h_2,\gv),(\yv,\pi_1,\pi_2,\ell)} = (\pi_1\circ
    \adww{h_1,\gv,\mv^1}) \oplus (\pi_2 \circ \adww{h_2,\gv,\mv^2}) \oplus
    \adww{\ell,\gv,\yv},
  \end{align*}
  and let
  \begin{align*}
    \F=\set{F_{u',v}\colon \cD \mapsto \cR}_{(u',v)\in \cU' \times\cV}.
  \end{align*}
  Finally, let \Bad be the set \Bad of~\cref{thm:ADWOurSetting}.

The above sets meet the requirements stated in \cref{lemma:GeneralRandPiClose}:  \cref{step:PPNotBad} of~\cref{thm:ADWOurSetting} assures that the first
property of~\cref{lemma:GeneralRandPiClose} is satisfied, and according to~\cref{step:PPProbBad} of~\cref{thm:ADWOurSetting} we set $\eps$ of~\cref{lemma:GeneralRandPiClose} to be $\const\cdot t/\size{\cU}^{k\cdot\cc}$, and
thus the second property is also satisfied. Hence, applying~\cref{lemma:GeneralRandPiClose} concludes the proof of the lemma.
\end{proof}

We note that for large enough $z$, and in contrast to the
\citet{PaghP08} family, using $\ADW$ we get meaningful results even when using an underlying
$k=o(\log t)$-wise independent family. (For a thorough comparison between $\PP$
and $\ADW$ see \cref{sec:pp_adw_comp}.) In particular, for specific
settings of parameters we get the following corollary.
\begin{corollary}\label{lemma:ADWUniStatClose}
  Let $t,\const,\cD,\cR,\cS,\cU, \cH,\cL,\cG$ be as in
  \cref{thm:ADWOurSetting}. Let $\Dc$ be an {\sf adaptive}, $t$-query
  oracle-aided algorithm and $c\in\N$. If either
  \begin{enumerate}
  \item $z=2(c+2)$, $\size{\cU}=t$ and $k=1$, or
  \item $z=2(c+2)\cdot\log t$, $\size{\cU} = 2$ and $k=1$,
  \end{enumerate}
  then
  \begin{align*}
    \size{\Pr_{f\la\ADW_\adwparam(\cH,\cL,\cG,\Pi_{\cS\mapsto\cR},\Pi_{\cU\mapsto\cS},\Pi_{\cU\mapsto\cR})}[\Dc^f =
      1] - \Pr_{\pi\la\Pi}[\Dc^\pi = 1]} \le \const/t^{c+1}.
  \end{align*}
\end{corollary}
\begin{proof}
  Since $\cH$ and $\cL$ are pairwise independent, we can apply
  \cref{lemma:ADWUniStatCloseGen}. Using either setting of
  parameters the advantage of $\Dc$ is
  $\const\cdot t/\size{\cU}^{\cc\cdot k/2} \leq \const/t^{c+1}$.
\end{proof}

\subsection{PRF Domain Extension  via the $\ADW$  Family}\label{sec:dom_ext_adw}
In this subsection we present a PRF a domain extension using the
\citet{AumullerDW14} family, $\ADW$. This allows us to avoid the large
independence required for using the \citet{PaghP08} family $\PP$.

In what follows we state two domain extension results (whose proofs are similar
to that of \cref{thm:PPDomainExtension}). In both of the results we
show how to extend the domain of a $(q,t,\eps)$-PRF
$\cF = \set{\cF_n\colon\zo^{s(n)}\mapsto\zo^{r(n)}}_{n\in\N}$. We let
$s,r,q,t,\eps,d,u$ be integer functions such that $q(n)\leq 2^{n-2}$, and let
$\cH= \set{\cH_n \colon \zo^{d(n)} \mapsto \zo^{s(n)}}_{n\in\N}$,
$\cL=\set{\cL_n\colon \zo^{d(n)}\mapsto \zo^{r(n)}}$ and
$\cG=\set{\cG_n\colon \zo^{d(n)}\mapsto \zo^{u(n)}}$ be efficient pairwise
independent function family ensembles. Let $p\in\poly$ be a large enough
polynomial determined by the evaluation and sampling time of $\cH$, $\cL$,
$\cG$, and $\cF$. Furthermore, both of the results will relay on
\cref{lemma:ADWUniStatClose} to argue that the function family
$\ADW_\adwparam(\cH,\cL,\cG,\Pi_{s,r},\Pi_{u,s},\Pi_{u,r})$
is indistinguishable from random (i.e., instantiate \cref{lemma:ADWUniStatClose}
with $\cS=\zo^s$, $\cR=\zo^r$ and $\cU=\zo^u$).

In the first result we rely on the first setting of paramters in
\cref{lemma:ADWUniStatClose} and implement the function families
$\Pi_{s,r},\Pi_{u,s}$ and $\Pi_{u,r}$ using a single pseudorandom function
family $\cF = \set{\cF_n\colon\zo^{s(n)}\mapsto\zo^{r(n)}}_{n\in\N}$. Assuming
$u(n)\leq s(n)$ and $s(n)\leq r(n)$, the implementation is done in the natural
way by padding with leading zeroes. Since
$\ADW_{\cc}(\cH, \cL, \cF, \cG, \cM, \cY)$ makes two calls to $\cF$, $2z$ calls
to $\cM$, and $z$ calls to $\cY$ (in total $3z+2$ calls), this results with a
PRF that makes $3z+2$ calls to the underlying PRF.

\begin{theorem}\label{thm:ADWDomainExtension}
  Let $\cH, \cL$ and $\cG$ be defined as above such that
  $u(n) \leq s(n) \leq r(n)$. Let $c>1$ be a constant and let $z= 2(c+2)$.

  Then, $\ADW_\adwparam(\cH,\cL,\cG,\cF,\cF,\cF) =
  \set{\ADW_\adwparam(\cH_n,\cL_n,\cG_n,\cF_n,\cF_n,\cF_n) \colon
    \zo^{d(n)}\mapsto\zo^{r(n)}}_{n\in\N}$ is a
  $(q,t-p\cdot q,(3z+2)\cdot\eps + 1/q^{c+1})$-PRF
\end{theorem}

In the second result we rely on the first setting of paramters in
\cref{lemma:ADWUniStatClose} and take advantage of the fact that $u=2$. This
enables us to implement the function families $\cM$ and $\cY$ as small tables of
random values which are embedded into the PRF key. This results with a PRF that
makes just two calls to the underlying PRF (but has a longer key).

\begin{theorem}\label{thm:ADWDomainExtension1}
  Let $\cH, \cL$ and $\cG$ be defined as above. Let $c>1$ be a constant and let
  $z=z(n)= 2(c+2)\cdot\log q(n)$. Let $\cM=\set{\cM_n}_{n\in\N}$ (\resp
  $\cY=\set{\cY_n}_{n\in\N}$) be family of tables, such that $\cM_n$ (\resp
  $\cY_n$) is a table of two (\resp $\adwparam(n)$) random elements from
  $\zo^{s(n)}$ (\resp $\zo^{r(n)}$).

Then, $\ADW_{\adwparam}(\cH,\cL,\cG,\cF,\cM,\cY) =
  \set{\ADW_{\adwparam(n)}(\cH_n,\cL_n,\cG_n,\cF_n,\cM_n,\cY_n) \colon
    \zo^{d(n)}\mapsto\zo^{r(n)}}_{n\in\N}$ is a
  $(q,t-p\cdot q,2\eps + 1/q^{c+1})$-PRF.
\end{theorem}

A notable difference between
\cref{thm:ADWDomainExtension,thm:ADWDomainExtension1} is that in
\cref{thm:ADWDomainExtension} the resulting PRF makes a large constant (\ie
$6c+12$) number of queries to the underlying PRF, whereas in
\cref{thm:ADWDomainExtension1} it only makes two calls to the underlying PRF but
the PRF key is longer (\ie it has $(6c+12)\cdot \log q$ random values embedded
into it).

\subsubsection{Comparing the $\PP$ and $\ADW$ Based Constructions}\label{sec:pp_adw_comp}
\cref{thm:PPDomainExtension,thm:ADWDomainExtension,thm:ADWDomainExtension1}
present different tradeoffs between two types of resources: ``cryptographic
work'' --- the total evaluation time of the calls to the underlying short-domain
PRF, and ``combinatorial work'' --- the independence needed from the function
families used.\footnote{The independence affects the amount of random bits and
  evaluation time needed for these families.} In the $\PP$-based construction
(\cref{thm:PPDomainExtension}), we minimize the number of calls to the PRF, thus
keeping the cryptographic work small. But on the other hand, we require
relatively high independence, which results in much combinatorial work. In
the first $\ADW$-based construction (\cref{thm:ADWDomainExtension}), the
situation is somewhat reversed: we minimize the independence needed, but make
more calls to the PRF. In the second $\ADW$-based construction
(\cref{thm:ADWDomainExtension1}), we minimize both the number of calls to the
PRF and the independence needed, but we require much more hash functions, thus
increasing again the combinatorial work.

In the following we instantiate
\cref{thm:PPDomainExtension,thm:ADWDomainExtension,thm:ADWDomainExtension1} in a
specific setting, and compare their performance in terms of required randomness
complexity, and evaluation time.\footnote{Another criterion of comparison is
  the sampling time. In our settings it is analogous to the evaluation time, so we omit it.} Our starting point is a length-preserving
$(q,t,\eps)$-PRF $\cF=\set{\cF_n\colon \zn\mapsto \zn}_{n\in\N}$. Our goal is to construct a
function family with a larger domain, $\zo^{d}$ for $d=d(n)$, whose security only
deteriorates, comparing to that of $\cF$, by an additive factor of $1/q^c$,
for $c\eqdef c(n)>1$. We consider the following function families:
\begin{enumerate}
\item\label{item:pp} \textbf{Family 1}: $\PP(\cH,\cF) = \set{\PP(\cH_n,\cF_n)\colon \zo^{d}
    \mapsto \zn}_{n\in \N}$, where  $\cH=\set{\cH_n \colon \zo^{d} \mapsto
    \zn}_{n\in\N}$ is $\Omega(c\cdot \log q)$-wise independent function family  ensemble.

  This instantiation, based on \cref{thm:PPDomainExtension}, makes only two calls to $\cF$, but requires hash functions
  of high independence.
\item\label{item:adw1} \textbf{Family 2}:
  $\ADW_\adwparam(\cH,\cH,\cH,\cF,\cF,\cF) =
  \set{\ADW_\adwparam(\cH_n,\cH_n,\cH_n,\cF_n,\cF_n,\cF_n) \colon
    \zo^{d}\mapsto\zo^{n}}_{n\in\N}$, where $z= 2(c+2)$ and
  $\cH= \set{\cH_n \colon \zo^{d} \mapsto \zo^{n}}_{n\in\N}$ is a
  pairwise independent function family ensemble.

  This instantiation, based on \cref{thm:ADWDomainExtension}, reduce the needed
  independence, in the price of increasing the number of calls to $\cF$.
\item\label{item:adw2} \textbf{Family 3}:
  $\ADW_\adwparam(\cH,\cH,\cH,\cF,\cM,\cY) =
  \set{ \ADW_{\adwparam(n)}(\cH_n,\cH_n,\cH_n,\cF_n,\cM_n,\cY_n) \colon
    \zo^{d}\mapsto\zo^{n}}_{n\in\N}$, where
  $z(n)=2(c+2)\cdot \log q$,
  $\cH= \set{\cH_n \colon \zo^{d} \mapsto \zo^{n}}_{n\in\N}$ is a
  pairwise independent function family ensemble, $\cM=\set{\cM_n}_{n\in\N}$
  (\resp $\cY=\set{\cY_n}_{n\in\N}$) is a family of tables, and  $\cM_n$ (\resp
  $\cY_n$) is a random table of two (\resp $\adwparam$) elements from $\zn$.

  This instantiation, based on \cref{thm:ADWDomainExtension1}, makes
  only two calls to $\cF$ and keeps the independence low, in the price of
  needing many more hash functions.
\end{enumerate}
To ease the comparison between the above families, we introduce relevant notation:
\begin{enumerate}
\item $\rand_\cF$ -- the amount of random bits required to sample a random
  element in $\cF_n$
\item $\eval_\cF$ -- the evaluation time of a single call to an element in
  $\cF_n$
\item $\evalManyWise{k}{d}{n}$ -- the evaluation time of a $k$-wise
  independent function family from $\zo^{d}$ to $\zn$. Recall that
  (\cref{fact:k-indep}), sampling a random element in the latter function family
  requires $k\cdot\max\set{d,n}$ random bits.
\end{enumerate}
\def\randPP{$O(c\cdot d\cdot \log q) + 2\cdot\rand_\cF$}
\def\randPPn{$O(c\cdot d(n)\cdot \log q(n)) + 2\cdot\rand_\cF(n)$}
\def\evalPP{$O\paren{\evalManyWise{c\cdot \log q}{d}{n}} + 2\cdot\eval_\cF$}
\def\evalPPn{$O\paren{\evalManyWise{c\cdot \log q(n)}{d(n)}{n}} + 2\cdot\eval_\cF(n)$}
\def\randADWone{$O(c\cdot d) + (6c+14)\cdot\rand_\cF$}
\def\randADWonen{$O(c\cdot d(n)) + (6c+14)\cdot\rand_\cF(n)$}
\def\evalADWone{$O\paren{c\cdot\evalManyWise{2}{d}{n}} + (6c+14)\cdot\eval_\cF$}
\def\evalADWonen{$O\paren{c\cdot\evalManyWise{2}{d(n)}{n}} + (6c+14)\cdot\eval_\cF(n)$}
\def\randADWtwo{$O(c\cdot d\cdot\log q) + 2\cdot\rand_\cF$}
\def\randADWtwon{$O(c\cdot d(n)\cdot\log q(n)) + 2\cdot\rand_\cF(n)$}
\def\evalADWtwo{$O\paren{c\cdot\log q\cdot \evalManyWise{2}{d}{n}} + 2\cdot\eval_\cF$}
\def\evalADWtwon{$O\paren{c\cdot\log q(n)\cdot \evalManyWise{2}{d(n)}{n}} + 2\cdot\eval_\cF(n)$}

The randomness complexity and evaluation time of each instantiation are
summarized in \cref{table1}.

\begin{table}
    \centering
    \begin{tabular}{ | l | l | l |}
      \hline
      Instantiation & Randomness Complexity (Key Size) & Evaluation Time \\ \hline
      $\PP$ (Family~\ref{item:pp})  & \randPP & \evalPP \\ \hline
      $\ADW$ (Family~\ref{item:adw1}) & \randADWone & \evalADWone \\ \hline
      $\ADW$ (Family~\ref{item:adw2})  & \randADWtwo & \evalADWtwo \\ \hline
    \end{tabular}
  \caption{Comparison between the domain extension results based on the function
    family $\PP$ from \cref{thm:PPDomainExtension} and  $\ADW$ from \cref{thm:ADWDomainExtension,thm:ADWDomainExtension1}  instantiated in \cref{item:pp,item:adw1,item:adw2} above
    to achieve security deterioration $1/q^c$ for $c\in \N$.}
    \label{table1}
\end{table}

\subsection{From Non-Adaptive to Adaptive PRF Via the $\ADW$ Family}\label{sec:nonadp_to_adp_adw}
In this subsection we state the result for the non-adaptive to adaptive
transformation using the \citet{AumullerDW14} family, $\ADW$. As in the previous
section, this allows us to avoid the large independence required for using the
function family $\PP$ of \citet{PaghP08}.  For simplicity we state only the
reduction the follows from the second settings of parameters of
\cref{lemma:ADWUniStatClose} (analogous to \cref{thm:PPNonAdaptiveMain}). The
first settings of parameters of \cref{lemma:ADWUniStatClose} also yields
non-adaptive to adaptive reduction, with different tradeoff between the
randomness complecity (key size) of the PRF to its evaluation time (see
\cref{sec:pp_adw_comp}).

Recall that our reduction from \cref{sec:NonadaptivetoAdaptive} requires that
all calls to the PRF $\cF$ are within the first $4q(n)$ first elements of
$\zn$. We make sure the above holds by the setting this to be the range of the
hash function we use.

\begin{theorem}\label{thm:ADWNonAdaptiveMain}
  Let $q$ be integer functions, let $c>1$ be a constant and let
  $z=z(n)= 2(c+2)\cdot\log q(n)$. Fix three efficient pairwise independent
  function family ensembles:
  $\cH= \set{\cH_n \colon \zo^{n} \mapsto [4q(n)]_\zn}_{n\in\N}$,
  $\cL=\set{\cL_n\colon \zn\mapsto \zn}_{n\in\N}$ and
  $\cG=\set{\cG_n\colon \zn\mapsto \zn}_{n\in\N}$.  Let
  $\cM=\set{\cM_n}_{n\in\N}$ (\resp $\cY=\set{\cY_n}_{n\in\N}$) be family of
  tables, such that $\cM_n$ (\resp $\cY_n$) is a random table of two (\resp
  $\adwparam(n)$) elements from $[4q(n)]_\zn$ (\resp $\zn$).

  If $\F = \set{\F_n\colon\zo^{n}\mapsto\zo^{n}}_{n\in\N}$ is a {\sf
    non-adaptive} $(4q,p\cdot t,\eps)$-PRF for some $p\in\poly$ determined by
  the evaluation time of $q,\cH,\cG$, $\cL$ and $\F$, then $\ADW_{\adwparam}(\cH,\cL,\cG,\cF,\cM,\cY) =
  \set{\ADW_{\adwparam(n)}(\cH_n,\cL_n,\cG_n,\cF_n,\cM_n,\cY_n) \colon
    \zn\mapsto\zn}_{n\in\N}$ is a
  $(q,t,2\eps + 1/q^{c+1})$-PRF.
\end{theorem}


%% file: OpenQuestions.tex
\section{Further Research}\label{sec:OpenQuestions}
The focus of this paper is on PRFs. Specifically, in \cref{sec:DomainExtension,sec:NonadaptivetoAdaptive,sec:InstantADW} we have shown domain extension techniques and
non-adaptive to adaptive transformations for PRFs that provide a nice tradeoff
between combinatorial work, cryptographic work and error. In general, hardness-preserving reductions between pseudorandom objects have led to fruitful research
with many results (some of which we review next). It is an interesting question
whether our technique has any bearing on other models.

Perhaps the most interesting model for this kind of reductions is pseudorandom
permutations (PRPs) (without going through a PRP-to-PRF reduction). Given a
family of $(q,t,\eps)$-PRPs from $n$-bits to $n$ bits, how can we construct a
family of PRPs with \textit{larger} domain while preserving its security?  How
about constructing a family of PRPs with \textit{smaller} domain while
preserving its security? Finally, it is also interesting how to transform a
family of $(q,t,\eps)$-PRPs that is secure against non-adaptive adversaries to a
family of PRPs that are also secure against adaptive adversaries (see discussion in \cref{sec:intro:RelatedWork}). One related
paper is that of \citet{HoangMR12}, that gives a method to convert a PRF into a
PRP with beyond-birthday security (see also
\cite{RistenpartY13,MorrisR14}). Other related works are of \citet{Hastad06} and \citet{MorrisRS09} showed how to extend the domain of a PRP.

A different model of interest is message authentication codes (MACs). In this model, we are interested in designing domain extension techniques that given an
$n$-bit to $n$-bit MAC with MAC security $\eps$ against $q$ queries provide variable-length MAC with some (good enough) promise on the MAC security in terms of $q$
and $\eps$. The best answer to-date for this question was given by \citet{DodisS11} that showed that given an $n$-bit to $n$-bit MAC with MAC security $\eps$
against $q$ queries, it is possible to get a variable-length MAC achieving MAC security $O(\eps \cdot q \cdot \poly (n))$ against queries of total length $qn$.

Another interesting model is public random functions. A public random
function $f:\zo^m\to\zo^n$ is a system with a public and private interface which
behaves as the same random function at both interfaces. In other words, a public
random function can be interpreted as a random oracle. In this model, again,
the domain extension problem is very interesting. To date, the best construction
is of \citet{MaurerT07} that presented a construction $\mathbf C_{\eps,m,\ell}$
that extends public random functions $f:\zo^n\to\zo^n$ to a function $\mathbf
C_{\eps,m,\ell}(f):\zo^{m(n)}\to\zo^{\ell(n)}$ with time complexity
$\poly(n,1/\eps)$ and which is secure against adversaries which make up to
$\Theta(2^{(1-\eps)n})$ queries.


On a different note, as we mentioned in the text following
\cref{thm:IntroNonAdaptiveMain}, our non-adaptive to adaptive transformation
depends on the number of queries made by the distinguisher which does not give a
single transformation for all poly-time adversaries. Whether a single reduction
with a constant number of calls to the non-adaptive PRF that works for all
poly-time adversaries exists is left as an open problem.

\citet{PatrascuT12} have shown that for many data structure problems it is
possible to use tabulation hashing even though it is `merely' 3-wise
independent. The question is whether this has any bearing on cryptographic
constructions.


%% file: NonAdaptiveFrameWork.tex
\newcommand{\emptystring}{\lambda}
\section{Proof of Lemma~\ref{lemma:GeneralRandPiClose}}\label{sec:GeneralFramework}
In this section we prove \cref{lemma:GeneralRandPiClose} from \cref{sec:GeneralFrameworkMain}. We mentioned again that \cref{lemma:GeneralRandPiClose} can be derived as  a special case of a result given in
\cite[Theorem~12]{JetchevOS12} (closing a gap in the proof appearing in \cite{Maurer02}). Yet, for the sake of completeness, we include an independent proof of this lemma below.

\remove{We provide a general framework for proving the \emph{adaptive} security of functions families, given that the analysis of the security of the \emph{non-adaptive} case behaves in a certain way: we need to be able to separate the `key' into two parts, drawn from $\cU$ and $\cV$ respectively. Then the analysis has to guarantee that  provided the \emph{non-adaptive} distinguisher did
not stumble upon a BAD set in $\cU$, the randomness of $\cV$ should make his view completely random. If the analysis is of this nature, then we claim that the resulting
construction is secure against adaptive attacks as well.}

\begin{definition}[Restating \cref{def:lhsmonotonesets}]
\LeftMonotoneSets
\end{definition}
\begin{lemma}[Restating \cref{lemma:GeneralRandPiClose}]\label{lemma:RandPiClose}
 \GeneralRandPiClose
\end{lemma}

\begin{proof}
Let $\mD$ be a $t$-query distinguisher. We assume for simplicity that $\mD$ is deterministic (the reduction to the randomized case is standard) and makes exactly
$t$ valid (\ie inside $\cD$) distinct queries. To prove the lemma we consider a process (\cref{alg:GenCoupling}) that runs $\mD$ twice: one giving it completely
random answers and the second time, choosing $u \la \cU$ and continuing answering with the same answers as the first round until we hit a BAD event according to
the queries and the chosen $u$. We then choose a random $v$ that is consistent
with the answers given so far and continue answering with it.

Intuitively, the answers provided to $\mD$ in the first round are distributed like the answers $\mD$ expects to get from a {\em truly random function}, while the answers provided to
$\mD$ in the second round are distributed like the answers $\mD$ expects to get from a random function in $\cF$. But, since these answers are the same until a BAD event occurs, the distinguishing ability of $\mD$ is bounded by the probability of such an event to occur. Since, $u$ is chosen \emph{after} $\mD$ has already ``committed'' to the queries it is going to make, this probability is bounded by the non-adaptive property of $\cF$.

For a vector $\vb = (v_1, \ldots, v_t)$, let $\vb_{1,\ldots,i}$ be the first $i$ element in $\vb$ (\ie
$\vb_{1,\ldots,i} = (v_1, \ldots, v_i)$) and let $\vb_{1,\ldots,0} = \emptystring$, where $\emptystring$ is the empty vector. Consider the following random process:
\begin{algorithm}\label{alg:GenCoupling}~
\begin{enumerate}
 \item \label{Drandom} Emulate $\mD$, while answering the $i^{\text{th}}$ query $q_i$ with $a_i \la \cR$.

 Set $\qb = (q_1, \ldots, q_t)$ and $\ab = (a_1, \ldots, a_t)$.
 \item\label{chooseU} Choose $u \la \cU$ and set $v = \perp$.

 \item\label{step:trivialCase} If $(\emptystring,u)\in\Bad$, set $v\la\cV$.

 \item\label{chooseV} Emulate $\mD$ again, while answering the $i^{\text{th}}$ query $q_i'$ according to the following procedure:
\begin{enumerate}

\item If $(\qb'_{1,\ldots,i} = (q_1',\dots,q_i'),u)\notin\Bad$, answer with $a'_i=a_i$ (the same $a_i$ from Step \ref{Drandom}).

\item Otherwise ($(\qb'_{1,\ldots,i},u)\in\Bad$):
\begin{enumerate}
\item \label{step:Bad} If $v = \perp$, set $v\la \set{v'\in \cV\colon \forall j \in [i-1] \colon f_{u,v'}(q'_j) = a'_j}$.

\item Answer with $a'_i = f_{u,v}(q'_i)$.

\end{enumerate}
\end{enumerate}

\item \label{step:last} Set $\overline{q'}=(q_1', \ldots, q_t')$ and $\overline{a'}=(a_1', \ldots, a_t')$. In case $v = \perp$, set $v\la \set{v'\in \cV\colon \forall j \in [t] \colon f_{u,v'}(q'_j) = a'_j}$.
\end{enumerate}
\end{algorithm}

Let $\Ab$, $\Qb$, $\Atb$, $\Qtb$, $U$ and $V$ be the (jointly distributed)
random variables induced by the values of $\ab$, $\qb$, $\overline{a'}$, $\overline{q'}$, 
$u$ and $v$ respectively, in a random execution of Algorithm \ref{alg:GenCoupling}. By definition  $\Ab$ has the same distribution as the oracle answers in a random
execution of $\mD^\pi$ with $\pi \la \Pi$. In \cref{claim:DistSame} we show that $\Atb$ is distributed the same as the oracle answers in a random execution of
$\mD^{f_{u,v}}$ with $(u,v) \la \cU\times\cV$. Using it, we now conclude the proof by bounding the statistical distance between $\Ab$ and $\Atb$.

Since the queries and answers in both emulations of $\mD$ at Algorithm \ref{alg:GenCoupling} are the same until $(\Qb_{1,\dots,i},U) \in\Bad$ for some $i\in [t]$,
and since $\Bad$ is monotone, it holds that
\begin{align}
\Pr[\Ab \neq \Atb] \leq \Pr[(\Qb,U)\in\Bad]
\end{align}
In addition, since $U$ is chosen \emph{after} $\Qb$, the second condition of \cref{lemma:RandPiClose} yields that
\begin{align}
\Pr[(\Qb,U)\in\Bad]\leq \eps
\end{align}
It follows that $\Pr[\Ab \neq \Atb] \leq \eps$ and therefore $\SD(\Ab,\Atb)\leq \eps$.

We conclude that
\begin{align*}
\size{\Pr_{\substack{u\la\cU \\ v\la\cV}}[\mD^{f_{u,v}} = 1] - \Pr_{\pi \la \Pi}[\mD^{\pi}=1]} \leq \SD(\Ab,\Atb) \leq \eps.
\end{align*}
\end{proof}

\begin{claim}\label{claim:DistSame}
$\Atb$ has the same distribution as the oracle answers in a random execution of $\mD^{f_{u,v}}$ with $(u,v) \la \cU\times\cV$.
\end{claim}
\begin{proof}
It is easy to verify that $\Atb$ is the oracle answers in $D^{f_{U,V}}$. Hence, to obtain the claim we need to  show that $(U,V)$ is uniformly distributed over $\cU
\times \cV$. The definition of Algorithm \ref{alg:GenCoupling} assures that $U$ is uniformly distributed over $\cU$, so it is left to show that conditioned on any
fixing $u$ of $U$, the value of $V$ is uniformly distributed over $\cV$.

In the following we condition on $U=u \in \cU$. For an answers vector $\wb\in\cR^k$, let $\qb_\wb$ [\resp $\qb_\wb^+$] be the first $k$ [\resp $k+1$] queries asked
by $\mD$, assuming that it gets $\wb$ as the first $k$ answers (since $\mD$ is deterministic these values are well defined). Let $\cS_\wb = \set{v \in \cV \colon
f_{u,v}(\qb_\wb) = \wb}$ and let $W = \set{\wb \in \cR^\ast \colon \size{\wb} \leq t \wedge (\qb_\wb,u) \notin \Bad}$. If $\emptystring \notin W$, it follows that
$(\emptystring ,u)\in\Bad$, and thus Algorithm \ref{alg:GenCoupling} chooses $v$ at Step \ref{step:trivialCase}. Hence $V$ is uniformly distributed over $\cV$. In
case $\emptystring \in W$, we conclude the proof by applying the following claim (proven below) with $\wb = \emptystring$ (note that $\cS_{\emptystring} = \cV$).

\begin{claim}\label{claim:UniInNode}
Conditioned on $\Atb_{1,\dots,i} = \wb \in W$ for some $i\in \set{0,,\dots, t}$, the value of $V$ is uniformly distributed over $\cS_\wb$.
\end{claim}
\end{proof}

\begin{proof}[Proof of \cref{claim:UniInNode}]
We prove by reverse induction on $i = \size{\wb}$. For the base case $i=t$, we note that (by definition) Algorithm \ref{alg:GenCoupling} chooses $v$ at Step \ref{step:last}, and thus $V$ is
uniformly distributed over $\cS_\wb$. In the following we assume the hypothesis holds for $i+1$, and condition on $\Atb_{1,\ldots,i} = \wb \in W$. In case $(\qb^{+}_\wb,u) \in \Bad$, Algorithm \ref{alg:GenCoupling} chooses $v$ at Step
\ref{step:Bad}, and thus $V$ is uniformly distributed over $\cS_\wb$. So it is left to handle the case $(\qb^{+}_\wb,u) \notin \Bad$.

Fix $v'\in\cS_\wb$ and let $a\in\cR$ be such that $v'\in\cS_{\wb\concat a}$, where `$\concat$' denotes vector concatenation (\ie for $\wb = (w_1, \ldots, w_i)$,
$\wb\concat a = (w_1, \ldots, w_i, a)$). Conditioning on $A'_{i+1}=a$, we can apply the induction hypothesis on $\wb\concat a$ (since $\wb\concat a \in W$) to get
that $V$ is uniformly distributed over $\cS_{\wb\concat a}$. It follows that
\begin{align*}
\Pr[V=v' \mid \Atb_{1,\ldots,i} = \wb] &= \Pr[A'_{i+1}=a\mid \Atb_{1,\ldots,i} = \wb]\cdot\Pr[v=v'\mid \Atb_{1,\ldots,i+1} = \wb \concat a] \\
&= \frac{1}{\size{R}} \cdot \frac{1}{\size{\cS_{\wb\concat a}}}\\
 &= \frac{1}{\size{R}} \cdot \frac{\size{\cR}^{\size{\wb} + 1}}{\size{\cV}}\\
 &= \frac{\size{\cR}^{\size{\wb}}}{\size{\cV}} = \frac{1}{\size{\cS_\wb}},
\end{align*}
concluding the induction step. The second equality holds by the induction hypothesis, and for the third one we note that
\begin{align}\label{eq:SWSize}
\frac{\size{\cS_{\wb'}}}{\size{\cV}} = \Pr_{v\la\cV}[v \in \cS_{\wb'}] = \Pr_{v\la\cV}[f_{u,v}(\qb_{\wb'}) = {\wb'}] = \frac{1}{\size{\cR}^{\size{{\wb'}}}},
\end{align}
for every ${\wb'}\in W$, where the third equality of \cref{eq:SWSize} holds by the first property of $\RandFunc$ (as stated in \cref{lemma:RandPiClose}).
\end{proof}

%% file: PRGtoPRF.tex
\section{Hardness-Preserving PRG to PRF Reductions}\label{sec:PRGImprove}
Another application of our technique is a hardness-preserving construction of PRFs from pseudorandom generators (PRGs). For instance, constructing $2^{c'n}$-PRF for some $0<c'<c$, from a $2^{cn}$-PRG. The efficiency of such
constructions is measured by the number of calls made to the underlying
PRG as well as other parameters such as representation size.

The construction of \citet{GoldreichGoMi86} (\ie \GGM) is in fact hardness preserving according to the above criterion. Their construction, however, makes $n$ calls to the underlying PRG,
which might be too expensive in some settings. 
\begin{proposition}[\cite{GoldreichGoMi86}]\label{prop:GGM}
Let $G$ be a length-doubling $(t,\eps)$-PRG whose evaluation time is
$e_G$. For any  efficiently-computable integer functions
$m$ and $\ell$, there exists an efficient oracle-aided function family ensemble whose $n$'th function family, denoted
$\GGM_{m(n)\rightarrow\ell(n)}^G$, maps strings of length $m(n)$ to strings of length $\ell(n)$, makes
$m(n)$ calls to $G$ and is a $(q, t - m \cdot q  \cdot e_G(\ell), m\cdot q \cdot \eps(\ell))$-PRF for any integer function $q$.\footnote{$\GGM_{m(n)\to\ell(n)}^G$ is a variant of the standard $\GGM$ function family, that on input of length $m(n)$ uses seed of length $\ell(n)$ for the underlying generator, rather than seed of length $m(n)$. Formally, $\GGM_{m\to\ell}^G$ is the function family ensemble $\set{\GGM^G_{m(n)\to\ell(n)}}_{n\in \N}$, where $\GGM^G_{m(n)\to\ell(n)} = \set{f_r}_{r\in \zo^{\ell(n)}}$, and for $r\in \zo^{\ell(n)}$, the oracle-aided function $f_r\colon \zo^{m(n)} \mapsto \zo^{\ell(n)}$ is defined as follows: given oracle access to a length-doubling function $G$ and input $x\in \zo^{m(n)}$,  $f_r^G(x) = r_x$, where $r_x$ is recursively defined by $r_\eps =r$,  and, for a string $w$, $r_{w||0} || r_{w||1} = G(r_w)$. The original $\GGM$ construction was length-preserving, \ie $m(n)=\ell(n)=n$.}\end{proposition}

As already mentioned in the introduction, in order to reduce the number of calls to the underlying PRG,
\citet{Levin87} suggested to first hash the input to a smaller domain, and only
then apply \GGM (this is known as ``Levin's trick''). The resulting
construction, however, is not hardness preserving due to the ``birthday attack''
described in \cref{sec:intro}.

While the \GGM construction seems optimal for the security it achieves (as shown in \cite{JainPT12}), in some settings the number of queries the distinguisher can
make is \emph{strictly less} than its running time. Consider a distinguisher of running time $2^{cn}$ that can only make $2^{\sqrt{n}} \ll 2^{cn}$
queries. In such settings the security of the \GGM construction seems like an overkill and raises the question of whether there exist more efficient reductions. \citet{JainPT12}
(who raised the above question) gave the following partial answer, by designing
a domain extension method tailored to PRG to PRF reductions for a specific range
of parameters.

\begin{theorem}[\cite{JainPT12}]\label{thm:JainInf}
  Let $G$ be a length-doubling $2^{cn}$-PRG.
Let $c>0$, $1/2\leq\alpha<1$ and $q(n)=2^{n^\alpha}$. There exists a length-preserving function family $\JPTCon^G$ that on input of length $n$ makes $O(\log(q(n)))=O(n^\alpha)$ calls to $G$ and is a $(q(n),2^{c'n},2^{c'n})$-PRF for every $0<c'<c$.
\end{theorem}
A restriction of \cref{thm:JainInf} is that it dictates that the resulting PRF
family makes \emph{at least} $\Omega(\sqrt{n})$ calls to the underlying PRG
(since $1/2\leq\alpha<1$). We note that the restriction that $\alpha>1/2$ (and hence $q(n)>
2^{\sqrt{n}}$) in the construction of \cite{JainPT12} is inherent due to their
hashing technique (and is not a mere side-effect of the parameters above). 

Using better hashing constructions (based on cuckoo hashing) yields a more versatile version of the above theorem, that  in particular allows $\alpha$ to be arbitrary. Specifically, combining \cref{prop:GGM} with \cref{thm:PPDomainExtension} yields the following result.
\begin{corollary}\label{cor:PPPRGImprov}
Let $G$ be a length-doubling $(t,\eps)$-PRG.
Let $\cH= \set{\cH_n \colon \zn \mapsto \zo^{m(n)}}_{n\in\N}$ and $\cG = \set{\cG_n \colon \zn \mapsto \zn}_{n\in\N}$ be efficient $k(n)$-wise independent function
family ensembles. Let $q(n)\leq
2^{m(n)-2}$. 

Then, the length-preserving function family ensemble $\set{\PP(\cH_n,\cG_n,\GGM^G_{m(n)\to n})}_{n\in\N}$, when invoked on input of length $n$, makes
$m(n)$ calls to $G$ and is a $(q, t - p \cdot m\cdot q ,2m\cdot q \cdot
\eps + q/2^{\Omega(k)})$-PRF, where $p(\cdot)$ is a polynomial determined by the evaluation and sampling time of $\cH$, $\cG$ and $G$.

In particular, for $c>0$, $0<\alpha<1$, $t(n)=2^{cn}$, $\eps(n)=1/t(n)$,
  $q(n) = 2^{n^\alpha}$, $m(n)=\Theta(\log (q(n)))$ and $k(n)=\Theta(n^\alpha + cn)$, the function family
  $\PP(\cH,\cG,\GGM^G_{m(n)\to n})$ makes $m(n)=O(n^\alpha)$ calls to $G$ and is a
  $(q(n),2^{c'n},2^{c'n})$-PRF, for every $0<c'<c$.
\end{corollary}
\begin{proof}
We prove the ``In particular'' part of the corollary. Set $q(n)=2^{n^\alpha}$, $m(n)=\lceil n^\alpha\rceil + 2$, and $\cH$ and $\cG$ to be $k(n)$-wise independent for $k(n)=\Theta(n^\alpha+cn)$, with an appropriate constant, such that $\frac{q(n)}{2^{\Omega(k(n))}} < 2^{-cn}$.
Let $t' = t - p \cdot m\cdot q$ and $\eps' = 2m\cdot q
\cdot \eps + q/2^{\Omega(k)}$. By the first part of the corollary we get that
$\PP(\cH,\cG,\GGM^G_{m(n)\to n})$ makes $m(n)$ calls to $G$ and is a
$(q,t',\eps')$-PRF. Next, we show that $t(n) > 2^{c'n}$ and $\eps(n)<2^{-c'n}$ for large enough $n$.

Let $c''\in\N$ such that $n^{c''} > p(n)$ for large enough $n$. It follows that $t(n) > 2^{cn} - n^{c''}2^{n^\alpha}(n^\alpha+2)$ and $\eps(n) < 2^{1+\log(n^\alpha+2)+n^\alpha-cn} + 2^{-cn}$. Hence, for every $c'<c$, we have $\eps(n)<2^{-c'n}$ and $t(n)>2^{c'n}$ for large enough $n$, as required.
\end{proof}

\paragraph{Comparison with the \citeauthor{JainPT12} reduction}
The advantage of \cref{cor:PPPRGImprov} is that when the adversaries are allowed to make less than $2^{\sqrt{n}}$ queries, the number of calls to the PRG is reduced accordingly, and below $O(\sqrt{n})$ calls. This improves upon the function family $\JPTCon$, that for such adversaries must make at least $O(\sqrt{n})$ calls to the PRG.

The function family \JPTCon, however, might have shorter description
(key) and evaluation time. Specifically, let $q$ denote the number of queries the adversaries are
allowed to make (\ie $q=2^{n^\alpha}$). According to \cref{cor:PPPRGImprov},
the parameter  $k$ (the independence required) needs to be set to
$\Theta(n)$. Hence, by \cref{fact:k-indep} it takes $\Theta(n^2)$ bits  to
describe a function in $\PP(\cH,\cG,\GGM^G_m)$. The evaluation time of a single
call to $\PP(\cH,\cG,\GGM^G_m)$ is $\Theta(\log q\cdot e_G + e_{(n)})$, where
$e_G$ is the evaluation time of $G$ and $e_{(k)}$ is the evaluation time of a $k$-wise independent function. In
comparison, it takes $\Theta(\log q\cdot n)$ bits to describe a member in
\JPTCon, and its evaluation time is $\Theta(\log q\cdot e_G
 + e_{(\log q)} )$. This is summarize in \cref{table2}.


\begin{table}
    \centering
    \begin{tabular}{ | l | l | l | l |}
      \hline
      Family & \#queries limitation & description (key) size & evaluation time \\ \hline
      $\JPTCon$ \cite{JainPT12} & $2^{n^{1/2}} < q < 2^n$ & $\Theta(\log q\cdot
      n)$ & $\Theta(\log q\cdot e_G + e_{(\log q)} )$ \\ \hline
      $\PP(\cH,\cG,\GGM^G_m)$ &  $0 < q < 2^n$ & $\Theta(n^2)$ & $\Theta(\log q\cdot e_G + e_{(n)})$ \\ \hline
    \end{tabular}
    \caption{Comparison between the family $\JPTCon$ in \cref{thm:JainInf}
      to $\PP(\cH,\cG,\GGM^G_m)$ in \cref{cor:PPPRGImprov}. The notation
      $e_{(k)}$ in the table refers to the evaluation time of a $k$-wise independent function.}
    \label{table2}
\end{table}

\paragraph{Independent work.} Independently and concurrently with this work, \citet{ChandranG14} showed that a
variant of the construction of \cite{JainPT12} achieves similar security
parameters to \cite{JainPT12} and also works for $2^{n^\alpha}$ queries for any
$0<\alpha<1/2$. The construction of \cite{ChandranG14}, however, outputs
only $n^{2\alpha}$ bits, as opposed to $n$ bits in the construction of
\cite{JainPT12} and in our construction.